%% file: main.tex
\documentclass[acmsmall,nonacm]{acmart}

\allowdisplaybreaks

\renewcommand\footnotetextcopyrightpermission[1]{} % no permission block
\setcopyright{none}

\makeatletter
\renewcommand\authorsaddresses{}
\gdef\@authorsaddresses{}
\makeatother

\AtBeginDocument{%
  }

\input{admin/preamble}

\begin{document}

\title{Selectivity Estimation for Linear Queries via Online Learning}


\author{Fangzhu Shen}
\email{fangzhu.shen@duke.edu}
\orcid{0009-0007-3820-5991}
\affiliation{%
 \institution{Duke University}
 \city{Durham}
 \state{North Carolina}
 \country{USA}}

\author{Debmalya Panigrahi}
\email{debmalya@cs.duke.edu}
\orcid{0000-0003-1799-6660}
\affiliation{%
 \institution{Duke University}
 \city{Durham}
 \state{North Carolina}
 \country{USA}}

\author{Sudeepa Roy}
\email{sudeepa@cs.duke.edu}
\orcid{0009-0002-8300-7891}
\affiliation{%
 \institution{Duke University}
 \city{Durham}
 \state{North Carolina}
 \country{USA}}

%\renewcommand{\shortauthors}{}

\cut{

       \cut{
       <concept_id>10003752.10010070.10010111.10003623</concept_id>
       <concept_desc>Theory of computation~Data provenance</concept_desc>
       <concept_significance>500</concept_significance>
       </concept>
       }

}


% \received{\red{20 February 2007}}
% \received[revised]{\red{12 March 2009}}
% \received[accepted]{\red{5 June 2009}}

\begin{abstract}
Learning-based approaches for selectivity estimation in databases have gained significant traction in recent years. However, theoretical studies of these learning-based approaches are essentially limited to fixed query distributions on static databases. 
In practice, both the underlying database and the query workload can dynamically change over time. 
In this work, we propose an algorithmic framework for learning selectivity of queries in this more general dynamic setup. Inspired by online learning, we measure the performance of the learning algorithm in this setting by its {\em regret}, which compares the cumulative loss incurred by the learning algorithm to that of the best fixed strategy. 
We establish upper and lower bounds on regret for histogram-based linear queries, such as point, range, and subset selection queries, under standard loss functions, in both static and dynamic database settings.
\end{abstract}

\maketitle

\input{sections_arxiv/s1-intro}

\input{sections_arxiv/static_squared_loss}

\input{sections_arxiv/dynamic_squared_loss}

\input{sections_arxiv/related_work}

\input{sections_arxiv/s6-conclusion}

\begin{acks}
%\section{Acknowledgments}
D. Panigrahi and F. Shen were supported in part by NSF grants CCF-1955703 and CCF-2329230.
 %This work was supported by the [...] Research Fund of [...] (Number [...]). Additional funding was provided by [...] and [...]. We also thank [...] for contributing [...].
\end{acks}

\bibliographystyle{ACM-Reference-Format}
\bibliography{admin/reference}

\clearpage
\appendix
\input{sections_arxiv/appedix/app_static_point}
\input{sections_arxiv/appedix/app_static_squared}
\input{sections_arxiv/appedix/app_dynamic_squared}

\input{sections_arxiv/appedix/app_dynamic_absolute}

\end{document}

%% file: admin/preamble.tex
\usepackage{algorithm}
\usepackage[noend]{algpseudocode}
\usepackage{amsmath}
\usepackage{amsthm}
\usepackage{bm}
\usepackage{amsfonts}
\usepackage{graphicx}
\usepackage{hyperref}
\usepackage{mathtools}
\usepackage{multicol}
\usepackage{tikz}
    \usetikzlibrary{arrows,shapes,positioning, automata}
    \usetikzlibrary{arrows,automata,positioning}
    \usetikzlibrary{decorations.markings}
    \usetikzlibrary{arrows.meta}
\usepackage[utf8]{inputenc}
\usepackage{xcolor}
\usepackage{xspace}
\usepackage[colorinlistoftodos]{todonotes}
\usepackage[shortlabels]{enumitem}
\usepackage{subcaption}
\usepackage[table]{xcolor}
\usepackage{hhline}
\usepackage{aliascnt}  % Required for cleveref to work with shared theorem counters
\usepackage{cleveref}
\usepackage{thmtools,thm-restate}

\usepackage{booktabs} % For nice tables
 \usepackage{multirow}

\newtheorem{theorem}{Theorem}[section]

\theoremstyle{plain}
\newtheorem{observation}{Observation}

\newcommand{\eat}[1]{}

\newtheoremstyle{cited}%
{.5\baselineskip\@plus.2\baselineskip
    \@minus.2\baselineskip}% (space above)
{.5\baselineskip\@plus.2\baselineskip
    \@minus.2\baselineskip}% (space below)
{\itshape}% (body font)
{\parindent}% (indent amount)
{}% {theorem head font}
{.}% {punctuation after theorem head}
{.5em}% {space after theorem head}
{\textsc{\thmname{#1}} \thmnote{\normalfont#3}}% {theorem head spec}

\theoremstyle{cited}

\newcommand{\paratitle}[1]{\smallskip\noindent\textbf{#1}~}
\newcommand{\cut}[1]{}
\newcommand{\red}[1]{{\color{red}#1}}

\newcommand{\R}{{\ensuremath{\mathbb R}}}

\newcommand{\Prob}{\mathbb{P}}

\newcommand{\commentout}[1]{{}}

\newcommand{\E}{\mathbb{E}}
\newcommand{\ones}{\mathbf{1}}
\newcommand{\zeros}{\mathbf{0}}
\newcommand{\evec}{\mathbf{e}}

\newcommand{\hvec}{\mathbf{h}}

\renewcommand{\k}{k} %number of bins
\renewcommand{\t}{t} %time step
\newcommand{\T}{T} %total time steps
\newcommand{\queryvec}{\mathbf{v}} %query indicator vector
\newcommand{\trueselect}{\sigma} %true selectivity
\newcommand{\predselect}{\hat{\sigma}} %predicted selectivity
\newcommand{\w}{\mathbf{w}} %weight vector
\newcommand{\predw}{\mathbf{\hat{w}}}
\newcommand{\wstar}{\mathbf{w}^*} %true weight vector
\newcommand{\prior}{\mu} %selectivity parameter
\newcommand{\regret}{\text{Regret}}
\newcommand{\loss}{\text{Loss}}
\newcommand{\inner}[2]{\langle #1, #2 \rangle}
\newcommand{\Hmat}{\mathbf{H}}
\newcommand{\maxent}{\textsc{SeqMaxEnt}}
\newcommand{\coordmem}{\textsc{CoordMem}}
\newcommand{\Yvec}{\mathbf{Y}}
\newcommand{\tildew}{\tilde{\mathbf{w}}} %unconstrained weight vector
\newcommand{\history}{\mathcal{H}} %history of query-selectivity pairs
\makeatletter
\@ifundefined{c@theorem}{}{}
\@ifundefined{c@lemma}{}{}
\@ifundefined{c@corollary}{}{}
\@ifundefined{c@proposition}{}{}
\@ifundefined{c@definition}{}{}
\@ifundefined{c@example}{}{}
\makeatother

%% file: sections_arxiv/s1-intro.tex
\section{Introduction}
\label{sec:intro}
Selectivity estimation is a core component of query optimization, enabling the query optimizer to choose efficient execution plans~\cite{ioannidis1991propagation,ioannidis2003history}. Traditional selectivity estimators often make strong assumptions such as independence among attributes on precomputed statistics (e.g., histograms, sketches, and samples), which may not hold in practice~\cite{poosala1997selectivity,cormode2012synopses,lipton1990practical,matias1998wavelet,poosala1996improved}. 
As a powerful alternative, a large number of learning-based approaches to selectivity estimation
have emerged in recent years and have exhibited good empirical performance~\cite{marcus2020bao,marcus2019neo,dutt2019selectivity,yang2020neurocard,park2020quicksel,sun_learned_2021,kim2022learned,hu2022selectivity,negi2023robust}. By leveraging machine learning models to capture complex data distributions and query patterns, these methods have potential for significantly improving accuracy compared to traditional methods.

Despite these strong empirical results, theoretical understanding of learning-based approaches for selectivity estimation remains limited.
Prior theoretical studies, such as the work by \citet{hu2022selectivity}, established that selectivity functions for range queries are learnable under the agnostic learning model~\cite{Haussler1992}, a generalization of the classical Probably Approximately Correct (PAC) framework~\cite{Valiant1984}. These results guarantee that, given enough training samples, a model's expected error will be small, if both training and future queries are drawn from the same fixed probability distribution.
Recent works have extended these guarantees to handle out-of-distribution queries and sequential data insertions~\cite{wu2024practical,zeighami2024theoretical,zeighami2024towards}.

These existing results are limited to static or mildly evolving datasets, and heavily rely on stochastic assumptions. %and largely focus on static or mildly evolving datasets.
Real-world database environments, however, are frequently dynamic: data is continuously inserted, updated, and deleted, and query patterns may shift unpredictably without following any stationary distribution. Under such dynamics, static models trained on historical data can become stale, leading to degraded accuracy or necessitating expensive retraining~\cite{wang2021ready}. While practical systems attempt to handle these dynamics through periodic retraining or lightweight adaptation~\cite{li2022warper,negi2023robust}, a rigorous theoretical framework that accommodates these fully dynamic, distribution-free settings would enable principled solutions. % for these dynamic settings.
%theoretical understanding of learning selectivity %estimation in such dynamic settings is % currently largely missing.

%\subsection*{Our Framework: Online Learning for Selectivity Estimation}
In this work, we present a theoretical study of selectivity estimation in the \emph{online learning} framework~\cite{shalev2012online,cesa2006prediction,hazan2016introduction}. This framework allows both data and incoming queries to change arbitrarily over time, and does not impose any distributional assumptions on either.
The online learning framework naturally mirrors the sequential nature of query processing, where the learner operates in sequential rounds and the prediction must be made before the true selectivity is revealed.
Concretely, in each round: (1) a query arrives; (2) the learner predicts its selectivity; (3) the true selectivity is revealed; then (4) the learner incurs a loss based on its prediction error.
The goal of the online learning algorithm is to minimize the cumulative loss across all rounds of this sequential process.

To formally evaluate a learning algorithm in our framework, we need to compare the learner's loss with some suitably defined benchmark. The standard benchmark used in online learning theory is the best fixed solution in hindsight. We adopt this principle and define our benchmark as the minimal cumulative loss incurred by the selectivity estimation corresponding to the best fixed database. Following standard terminology, we call the difference between the algorithm's loss and this benchmark the \textbf{regret} of the algorithm. Our choice of benchmark captures the intuition that a meaningful learning guarantee requires the problem instances in the sequential rounds to share some underlying structure. In contrast, if we allow our benchmark to correspond to a dynamic database, then the optimal strategy would constitute unrelated, separate solutions for the queries in different rounds. In this situation, it is clear that no algorithm can make meaningful predictions for the next query based on its knowledge of previous queries. Indeed, we formally show in Section~\ref{sec:regret} that no online algorithm can achieve non-trivial regret guarantees against a dynamic benchmark in this case. Comparing against the best fixed database isolates the learner's ability to discover and exploit persistent structure in the workload; a sublinear regret bound guarantees that the learner's average loss approaches that of the optimal static strategy, without any distributional assumptions required on queries or data.

\subsection{Database and Query Models}
\label{sec:database-query-models}
We consider the standard histogram representation of a database, in which the data domain has a discrete support of size $\k$ or is partitioned into $\k$ bins. (We refer to $k$ as the support size of the domain.) The state of the database is represented by a vector $\w \in \R^\k_{\geq 0}$, where the $i$-th coordinate $w_i$ represents the relative frequency of the $i$-th element (or $i$-th bin). The set of all valid database vectors forms the $\k$-dimensional simplex:
%Let $\dataset$ be a database with a discrete and finite support $\{x_1,\dots,x_\k\}$, where $x_i$ is the $i$-th element of the support.  
%\begin{equation}\label{eq:database-vector}
$\Delta_\k := \left\{ \w \in \R^\k \;\middle|\; \sum_{i=1}^\k w_i = 1, \; w_i \ge 0 \; \forall i \in [\k] \right\}$.
%\end{equation}

We study \emph{linear queries}, represented by $k$-dimensional vectors $\queryvec \in [0,1]^\k$, whose {\em selectivity} $\sigma(\queryvec, \w)$ for a database $\w$ is given by the weighted sum $\sigma(\queryvec, \w) = \langle \queryvec, \w \rangle = \sum_{i=1}^\k v_i w_i$.
We note that linear queries include commonly studied query classes:
\begin{itemize}
    \item[-] \emph{Subset queries} $\mathcal{Q}_{\text{subset}}:= \{\{0,1\}^\k\}$, whose selectivity captures the cumulative relative frequency of a set of elements.
%A subset query is a special case of a linear query where the coverage is binary. It is represented by a vector $\queryvec \in \{0,1\}^\k$, where $v_i = 1$ if $x_i$ belongs to the subset and $v_i = 0$ otherwise. The class of subset queries is given by: $\mathcal{Q}_{\text{subset}} = \{0, 1\}^{\k}$.
    \item[-] \emph{Range queries} 
    $\mathcal{Q}_{\text{range}}:= \left\{\queryvec \in \{0,1\}^\k \;\middle|\; v_i = 1 \text{ if } a \le i \le b \text{ for some } a \le b\right\}$, 
    which is a special case of subset queries where the set of indices corresponds to a range of (ordered) elements.
    %When the support set of a database is ordered (e.g., partitioning a 1-dimensional continuous domain or when the discrete domain has a natural order), the range query asks for the selectivity of a contiguous interval of the support.
%Given an interval $[a,b]$, its query vector has $v_i = 1$ for $i \in [a, b]$ and $v_i = 0$ otherwise. 
    \item[-] \emph{Point queries}
     $\mathcal{Q}_{\text{point}}:= \left\{ \queryvec \in \{0,1\}^\k \;\middle|\; v_i = 1 \text{ if } i = a \text{ for some } a \right\}$, which further refines range queries and isolates the relative frequency of a specific element.
\end{itemize}
As described above, these query classes are nested as $\mathcal{Q}_{\text{point}} \subset \mathcal{Q}_{\text{range}} \subset \mathcal{Q}_{\text{subset}} \subset \mathcal{Q}_{\text{linear}}$. %\red{The upper bounds hold for broad class also applies to the restricted subclasses, and the lower bounds hold for the restricted subclasses also applies to the broad class. }

\subsection{The Online Learning Framework: Loss Functions and Regret}\label{sec:regret}
The selectivity estimation problem asks a learner to predict the selectivities of a sequence of $T$ queries from a given class $\mathcal{Q}$ based only on prior history. We denote the $\t$-th query by $\queryvec_{\t}$ and its true selectivity by $\trueselect_{\t}$. In round $\t = 1, \dots, \T$, we have the following steps:
\begin{enumerate}
    \item A query vector $\queryvec_{\t} \in \mathcal{Q}$ is presented to the learner (algorithm $\mathcal{A}$). 
    \item Note that at this stage, the information available to algorithm $\mathcal{A}$ is given by $\history_{\t}:= (\queryvec_1, \trueselect_1, \queryvec_2, \trueselect_2, \ldots, \queryvec_{\t-1}, \trueselect_{\t-1}, \queryvec_{\t})$, i.e., all the previous queries and their respective selectivities as well as the new query (but not its selectivity). Based on this information, algorithm $\mathcal{A}$ outputs a selectivity prediction $\predselect_{\t} \in [0,1]$.
    \item Next, the true selectivity $\trueselect_{\t}$ of query $\queryvec_{\t}$ is revealed.
    Note that $\trueselect_{\t}$ corresponds to a (possibly time-varying) unknown database vector $\w_{\t} \in \Delta_\k$, i.e., $\trueselect_{\t} = \sigma(\queryvec_{\t}, \w_{\t}) = \langle \queryvec_{\t}, \w_{\t} \rangle$.
    %We assume there exists a (possibly time-varying) unknown true database state $\w_{\t} \in \Delta_\k$ such that $\trueselect_{\t} = \sigma(\queryvec_{\t}, \w_{\t}) = \langle \queryvec_{\t}, \w_{\t} \rangle$.
    The learner observes only $\trueselect_{\t}$, but not the underlying $\w_{\t}$.
    \item The learner incurs a loss based on the difference between the true and predicted selectivities: $\loss_{\t} := \loss(\trueselect_{\t},\predselect_{\t})$.
We consider two natural loss functions: 
\begin{itemize}
\itemsep0em
    \item[-] \emph{Squared loss}: \ \ \ \ \ \ $\loss_{\t} = (\trueselect_{\t} - \predselect_{\t})^2$
    \item[-] \emph{Absolute loss}:  \ \ \ \ \  $\loss_{\t} = |\trueselect_{\t} - \predselect_{\t}|$. 
\end{itemize}
%They capture different practical database requirements.
Squared loss, which heavily penalizes large estimation errors, is useful in query optimization applications where a single large error may lead to a poorly optimized execution plan, whereas
absolute loss captures average error, which is robust to isolated outliers. %This is useful when occasional anomalous queries are expected and should not dominate the learned model. 
%As our results show, these two loss functions lead to different regrets.
\end{enumerate}

\smallskip
\noindent
{\bf Regret.~}
We measure the performance of a learning algorithm $\mathcal{A}$ using the notion of \emph{regret}~\cite{hazan2016introduction,shalev2012online}, denoted $\regret_{\T}(\mathcal{A})$. The regret of the algorithm quantifies the difference between the cumulative loss of the algorithm and that of the best fixed comparator (i.e., the best fixed database state) chosen in hindsight. 
For a fixed comparator $\w \in \Delta_\k$, it predicts the selectivity at rounds $\t$ as: $\sigma(\queryvec_{\t}, \w) = \langle \queryvec_{\t}, \w \rangle$, thus we denote its round-$t$ loss as $\loss_\t (\w) := \loss(\trueselect_{\t}, \langle \queryvec_{\t}, \w \rangle)$. Therefore, the regret is defined as:
\begin{align}\label{eq:regret}
    \regret_{\T}(\mathcal{A}) = \sum_{\t=1}^{\T} \loss_{\t} - \min_{\w \in \Delta_\k} \sum_{\t=1}^{\T} \loss_\t (\w).
\end{align}
Equivalently,
\begin{align}\label{eq:regret2}
    \regret_{\T}(\mathcal{A}) = \sum_{\t=1}^{\T} \loss(\trueselect_{\t}, \predselect_{\t}) - \min_{\w \in \Delta_\k} \sum_{\t=1}^{\T} \loss(\trueselect_{\t}, \langle \queryvec_{\t}, \w \rangle).
\end{align}

Our goal is to design algorithms that minimize regret.

\smallskip
\noindent
{\bf Why compare with a fixed solution for regret?}~
It is standard practice in online learning to compare the algorithm's loss to that of the best {\em fixed} solution in hindsight~\cite{hazan2016introduction,shalev2012online,cesa2006prediction} as given above in \eqref{eq:regret} and \eqref{eq:regret2}. In our setting, this fixed solution is a single database vector that incurs minimal loss after seeing the entire sequence, i.e., with $\min_{\w \in \Delta_\k} \sum_{\t=1}^{\T} \loss_\t(\w)$. This is in contrast to comparing with a stronger benchmark that would allow a different solution in every round, i.e., $\sum_{\t=1}^{\T} \min_{\w_t \in \Delta_\k} \loss_\t(\w_t)$. In general, no algorithm, even a randomized one, can obtain sub-linear (in $T$) regret against this dynamic benchmark as illustrated in the following example. 
\begin{example}
    Consider an instance where the database vector $\w_\t$ in each round is either $(1, 0)$ or $(0, 1)$, and the query is always $\queryvec_\t=(1, 0)$. For any (possibly randomized) algorithm, the adversary chooses the database $\w_{\t}$ to be $(1,0)$ whenever the algorithm outputs a selectivity prediction $\predselect_\t < 1/2$, and $(0, 1)$ otherwise. That is, the true selectivity $\trueselect_\t = 1$ if $\predselect_\t < 1/2$ and $\trueselect_\t = 0$ if $\predselect_\t \geq 1/2$. Therefore, the cumulative squared and absolute loss are respectively $\T/4$ and $\T/2$. In contrast, the dynamic benchmark can choose the realized database vector in each round, so its cumulative loss is $0$. Therefore, the expected regret against this dynamic benchmark is linear in $\T$, even for randomized learners.    
    % that outputs a selectivity prediction $\predselect_\t \in [0,1]$, its expected loss at round $\t$ is:
    % \begin{equation*}
    %     \mathbb{E}\left[ \loss_{\t} \right] = 
    %     \begin{cases}
    %         \mathbb{E}\!\left[|\trueselect_\t-\predselect_\t|\right] = \frac{1}{2} (|\predselect_\t| + |1 - \predselect_\t|) = \frac{1}{2}, & \text{under absolute loss}, \\
    %         \mathbb{E}\!\left[(\trueselect_\t-\predselect_\t)^2\right] = \frac{1}{2}\predselect_\t^2+\frac{1}{2}(1-\predselect_\t)^2 \geq \frac{1}{4}, & \text{under squared loss}, \\
    %     \end{cases}
    % \end{equation*}
\end{example}

\subsection{Static vs.\ Dynamic Databases}\label{subsec:static-dynamic-compare}

We analyze the framework in two database settings that exhibit qualitatively different behavior:
\begin{itemize}
    \item[-] {\em Static databases.}
In this setting, there is a fixed but unknown database vector $\w^* \in \Delta_\k$ such that $\w_\t = \w^*$ in every round, while queries may still be chosen arbitrarily. This models environments such as analytical warehouses, where the underlying data is updated infrequently in batches, yet the query workload may shift over time as analysts explore different aspects of the data.
    \item[-] {\em Dynamic databases.}
The database state $\w_\t$ may change arbitrarily across rounds, in addition to the queries. This models transactional or production systems with continuous insertions, updates, and deletions, where the learner must simultaneously adapt to a moving database and an evolving query workload.
\end{itemize}
%We studythe regret in two database settings that exhibit qualitatively different regret behavior. In the \emph{static database} setting, there exists a fixed but unknown database vector $\w^* \in \Delta_\k$ suct that $\w_\t = \w^*$ in every round, while queries may still be chosen arbitrarily; this models analytical warehouses where the data changes slowly relative to the query workload. In the \emph{dynamic database} setting, the database $\w_\t$ may also change arbitrarily across rounds; this models transactional systems with continuous insertions, updates, deletions, and shifting query patterns occurring simultaneously. The following observation captures the qualitative gap between the two:

The following observation captures the difference and provides a baseline for the static setting:
\begin{observation}\label{obs:static-baseline}
    For static databases, the regret of any algorithm $\mathcal{A}$ equals its cumulative loss,
    \begin{equation}\label{eq:regret-static}
        \regret_{\T}(\mathcal{A}) = \sum_{\t=1}^{\T} \loss_\t,
    \end{equation}
    and depends only on the support size $\k$, not on the time horizon $\T$. In particular, for any linear query and either squared or absolute loss, there exists an online algorithm with regret at most $\k-1$.
    % (both upper bound for any reasonable algorithm and the lower bound) 
\end{observation}

\begin{proof}
Since $\w^*$ is fixed, the best fixed comparator in hindsight chooses $\w^*$ and incurs zero cumulative loss; thus by \eqref{eq:regret}, the regret of any algorithm $\mathcal{A}$ equals its cumulative loss. Moreover, each revealed selectivity $\trueselect_{\t} = \langle \queryvec_{\t}, \w^* \rangle$ imposes a linear constraint on $\w^*$. Since $\w^*$ has $k$ unknown components and is normalized, after at most $\k - 1$ linearly independent queries $\w^*$ can be uniquely determined; subsequent queries incur zero loss. For each round before $\w^*$ is identified, the query is either linearly dependent on the previously observed queries, in which case the prediction is exact and the loss is zero, or is linearly independent, in which case it incurs loss at most $1$. Hence the 
total regret is at most $\k - 1$.
\end{proof}

For static databases, our results (\Cref{tab:static_results}) demonstrate regret bounds much sharper than the $\k-1$ baseline of \Cref{obs:static-baseline}, and remain independent of the time horizon $\T$. For dynamic databases, since the best fixed comparator may itself incur loss that grows with $\T$, a dependence on $\T$ becomes unavoidable (\Cref{tab:dynamic_results}). We therefore treat $\T$, rather than the support size $\k$, as the dominant parameter in the dynamic setting and analyze how regret scales in $\T$.

\subsection{Overview of Results}
\label{sec:our-results}
We provide a systematic regret analysis for selectivity estimation of linear queries across different database settings and loss functions, which we summarize in Table~\ref{tab:static_results} and Table~\ref{tab:dynamic_results}. 
%When multiple query classes appear in the same row, the upper and lower bounds apply to all of them. %\blue{Our results show that the achievable regret depends on the database setting, the loss function, and the query class. For some narrow classes such as point queries, we obtain sharper bounds than those obtained by viewing them as general linear queries, reflecting the fact that more structured queries reveal information about the database more efficiently.}

\begin{table}[t]
    \caption{Regret bounds for selectivity estimation on \textbf{static} databases. $\k$ is the support size of the domain. Each row reports upper and lower regret bounds holding for the indicated query classes.\cut{When multiple query classes are listed in the same row, both the upper and lower bound apply to all of them.}}
    \label{tab:static_results}
    \begin{center}
    \vspace{-0.4cm}
    \renewcommand{\arraystretch}{1.4}
    \footnotesize
    \begin{tabular}{|c|c|c|c|c|}
        \hline
        \textbf{Loss} & \textbf{Query Class} & \textbf{Lower Bound} & \textbf{Upper Bound} & \textbf{Algorithm} \\
        \hline
        \multirow{4}{*}{Squared}
            & \multirow{2}{*}{Point} & $1/4$ & $1$ & \coordmem \\
        \cline{3-4}
            & & \multicolumn{2}{c|}{(\Cref{prop:static-point})} & (\Cref{alg:coord-mem}) \\
        \hhline{|~|-|-|-|-|}
            & Range, Subset, Linear & $\Omega(\log k)$ & \cellcolor{gray!15}$O(\log k)$ & \cellcolor{gray!15}\maxent \\
            & & (\Cref{thm:static_lower}) & \cellcolor{gray!15}(\Cref{thm:static_upper}) & \cellcolor{gray!15}(\Cref{alg:seq_max_ent}) \\
        \hline
        \multirow{4}{*}{Absolute}
            & \multirow{2}{*}{Point} & $1/2$ & $1$ & \coordmem \\
        \cline{3-4}
            & & \multicolumn{2}{c|}{(\Cref{prop:static-point})} & (\Cref{alg:coord-mem}) \\
        \hhline{|~|-|-|-|-|}
            & Subset, Linear
            & \cellcolor{gray!15}$\Omega\!\left(\sqrt{k/\log k}\right)$
            & $O\!\left(\sqrt{k \log k}\right)$
            & \maxent\ \\
            & & \cellcolor{gray!15}(\Cref{thm:static_absolute_lower_bound}) & (\Cref{thm:static_absolute_upper}) & (\Cref{alg:seq_max_ent}) \\
        \hline
    \end{tabular}
    \end{center}
\end{table}

\begin{table}[t]
    \caption{Regret bounds for selectivity estimation on \textbf{dynamic} databases. $\T$ is the number of rounds and $\k$ is the support size of the domain. Each row reports upper and lower regret bounds holding for the indicated query classes.\cut{When multiple query classes are listed in the same row, both the upper and lower bound apply to all of them.}}
    \label{tab:dynamic_results}
    \begin{center}
    \vspace{-0.4cm}
    \renewcommand{\arraystretch}{1.4}
    \footnotesize
    \begin{tabular}{|c|c|c|c|c|}
        \hline
        \textbf{Loss} & \textbf{Query Class} & \textbf{Lower Bound} & \textbf{Upper Bound} & \textbf{Algorithm} \\
        \hline
        \multirow{2}{*}{Squared}
            & Point, Range, Subset, Linear
            & \cellcolor{gray!15} $\Omega(\log T)$
            & $O(k \log T)$
            & EWOO~\cite{hazan2007logarithmic} \\
            & & \cellcolor{gray!15}(\Cref{thm:dynamic_lower}) & (\Cref{cor:ewoo_log_regret}) & (\Cref{alg:ewoo}) \\
        \hline
        \multirow{4}{*}{Absolute}
        & Point
        & $\Omega(\sqrt{T})$
        & $O(\sqrt{T})$
        & FTRL with $\ell_2$ regularizer~\cite{shalev2012online} \\
        & & (\Cref{thm:dynamic_lower_absolute}) & (\Cref{cor:point-query-upper}) & (\Cref{alg:ftrl-abs})\\
        \cline{2-5}
            & Subset, Linear
            & $\Omega\!\left(\sqrt{T \log k}\right)$
            & $O\!\left(\sqrt{T \log k}\right)$
            & FTRL with neg-entropy~\cite{shalev2012online} \\
            & & (\Cref{thm:dynamic_lower_k_bins}) & (\Cref{cor:dynamic-abs-upper}) & (\Cref{alg:ftrl-abs}) \\
        \hline
    \end{tabular}
    \end{center}
\end{table}

Across all settings, our regret bounds are (nearly) tight in the dominant parameter. For static databases, we establish $\Theta(\log \k)$ regret under squared loss for range, subset, and linear queries, and $\tilde{\Theta}(\sqrt{\k})$\footnote{$\tilde{\Theta}(\cdot)$ is an asymptotic 
notation that ignores poly-logarithmic factors.} regret under absolute loss for subset and linear queries. For dynamic databases, where the dependence on $\T$ is dominant, we obtain $\Theta(\sqrt{\T})$ and $\Theta(\sqrt{\T \log \k})$ regret under absolute loss for point and subset/linear queries respectively, and $\Theta(\log \T)$ regret in $\T$ under squared loss. Thus the two loss functions exhibit a polynomial separation in the dominant parameter under both database settings.
%The choice of loss function thus exhibits a polynomial separation in the dominant parameter under both database settings.

\subsection{Overview of Techniques}

Of the results summarized above, we highlight the techniques used in obtaining the three results in shaded boxes in \Cref{tab:dynamic_results,tab:static_results}. We also present these results in detail in the main body of the paper, and defer the details of the remaining results to the appendix because of space constraints.
\begin{enumerate}[leftmargin=*]
    \sloppy
    \item \textbf{The Sequential Maximum Entropy (\maxent) algorithm for static databases (\Cref{alg:seq_max_ent} and \Cref{thm:static_upper}).} We design the \maxent\ algorithm that achieves $O(\log k)$ regret under squared loss and $O(\sqrt{k \log k})$ regret under absolute loss for general linear queries. This algorithm maintains a feasible space of candidate databases that are consistent with the known selectivities of all previous queries. As outlined in the proof of \Cref{obs:static-baseline}, selecting a database arbitrarily in the feasible space in each step and outputting its selectivity as the prediction for a given query already gives a cumulative loss of $k-1$. 
    
    To improve on this strategy, the \maxent\ algorithm uses a more deliberate choice of the candidate database. In particular, it outputs the selectivity of the {\em maximum entropy} vector in the feasible space as its prediction for a given query. Intuitively, this ensures that if the algorithm suffers a large loss for a given query, then the correct database vector must be ``far'' from the centrally located maximum entropy vector in the feasible space. This, in turn, ensures that the algorithm makes significant progress toward identifying the true underlying database vector.
    %Because the database is fixed, each query-selectivity observation adds a linear constraint and the feasible region shrinks monotonically across rounds. Therefore, we use this structural property that generic online convex optimization on the simplex cannot exploit, to achieve $T$-independent regret. 
    This is formalized for squared loss (\Cref{thm:static_upper}) by showing that each round's loss is bounded by an incremental KL divergence between successive predictions, which telescopes across rounds to at most $O(\log k)$. This bound is then used to derive a bound of $O(\sqrt{k \log k})$ for absolute loss (\Cref{thm:static_absolute_upper}) using a standard relationship between the squared and absolute loss functions (\Cref{obs:squared-to-absolute}).

    \item \textbf{A $\Omega(\sqrt{k/\log k})$ lower bound for static databases under absolute loss (\Cref{thm:static_absolute_lower_bound}).} %We prove that any algorithm incurs $\Omega(\sqrt{k/\log k})$ regret on static databases under absolute loss, even for subset queries.
    A natural strategy for proving a lower bound is to make the learner face a fresh independent $\mathrm{Bernoulli}(1/2)$ variable in each round. Since the algorithm has no information about the current round's random variable from previous observations, it can do no better than random guessing and must incur expected loss of $1/2$ per round.
    In the dynamic setting, this is straightforward since the database can change across rounds and introduce new independent randomness directly (see, e.g., \Cref{thm:dynamic_lower_absolute}). However, the static setting only allows a single fixed database over all rounds. The challenge is therefore to encode many independent random variables into one fixed database vector such that the query sequence reveals the random outcomes piecemeal over time. 
    %More importantly, the resulting database vector must be valid, i.e., lie in the simplex.

    We achieve this encoding via a normalized Hadamard matrix~\cite{horadam2012hadamard} of order $\Theta(\k)$, whose rows are mutually orthogonal and whose first row consists entirely of $+1$s. Starting from the uniform vector, we perturb it along the direction of each Hadamard row by an independent random perturbation of some fixed magnitude $\delta$. Each round's subset query is built from one row of the same Hadamard matrix so that the true selectivity per round depends only on one random perturbation.
    %, which is achieved by the orthogonality and normalization of the Hadamard matrix. Therefore, 
    This construction forces the learner to predict a perturbation in each round and incur $\Omega(\delta)$ expected absolute loss. 
    %$\delta$ is a magnitude of the perturbation, which is a constant to be determined later.

    But the above construction is not always valid: adding random signs may cause some coordinate of the database vector to be negative in some realizations and fall outside the simplex. This exposes a tradeoff between two competing demands: a larger $\delta$ amplifies the per-round loss, while a smaller $\delta$ keeps the perturbed vector inside the simplex with high probability. By choosing $\delta = \Theta(1/\sqrt{k\log k})$ to balance the two and summing over $\Theta(\k)$ rounds, we obtain the lower bound of $\Omega(\sqrt{k/\log k})$.

    \item \textbf{A $\Omega(\log T)$ lower bound for dynamic databases under squared loss (\Cref{thm:dynamic_lower}).}
%    Even for point queries on a $2$-bin database (i.e., $\k=2$), any algorithm incurs $\Omega(\sqrt{T})$ regret.
%    Following the independent-randomness strategy above, one may consider generating 
%   
    Suppose the true selectivity in every round is given by an independent $\mathrm{Bernoulli}(1/2)$ variable.
    This can be achieved in the dynamic setting even for $k=2$ with a fixed query by randomizing the database independently in each round. In this case, by standard variance bounds, we get a regret lower bound of
    $\Omega(\sqrt{T})$ for absolute loss, as desired. However, this construction yields only a constant bound for squared loss: every algorithm incurs expected cumulative loss $\T/2$, while the best fixed comparator now chooses the empirical mean in hindsight, improving by a constant $1/4$ compared to the algorithm. 
    %of true selectivities after all rounds,
    %Thus, by the regret definition \eqref{eq:regret}, we can only get a constant regret bound. 
    
    %The idea gives the right lower bound for absolute loss in the dynamic setting: every learner incurs expected loss $\T/2$, while the best fixed comparator can choose the majority value in hindsight and improves by the deviation of number of ones from $\T/2$, giving $\Omega(\sqrt{T})$ lower bound on the regret.
    %However, the same construction is too weak under squared loss: The best fixed comparator now chooses the empirical mean $H/T$, where $H$ is the number of ones, and its improvement is only $(H-T/2)^2/T$, whose expectation is constant. Thus, a squared-loss lower bound cannot rely on the sample imbalance of independent fair coin.
    Instead of making the algorithm guess a $\mathrm{Bernoulli}(1/2)$ variable in each round, we draw a parameter $\prior \sim \mathrm{Uniform}[0,1]$ once at the start, and randomize the database to generate each round's selectivity as an i.i.d.\ $\mathrm{Bernoulli}(\prior)$ sample. In this case, the expected regret in each round $\t$ is lower bounded by the posterior variance of $\prior$,
    which follows a Beta distribution with expected variance $\Theta(1/\t)$. Thus, the regret lower bound becomes the Harmonic sum over the $\T$ rounds, i.e., $\Omega(\log \T)$.
\end{enumerate}

\smallskip
\noindent\textbf{Summary of Contributions.~}
Overall, our work provides distribution-free, worst-case theoretical guarantees for learned selectivity estimation, and connects selectivity estimation for linear query classes with online learning theory. The main contributions are: (i) an online-learning formulation of selectivity estimation that imposes no distributional assumptions on data or queries, and (ii) (nearly) tight regret bounds for %broad 
common query classes, including point, range, subset, and linear queries, across static and dynamic database settings, under squared and absolute loss.

%% file: sections_arxiv/static_squared_loss.tex
\section{Static Databases}
\label{sec:static}
%equals its cumulative loss (Eq.~\ref{eq:regret-static}) and 
As established in \Cref{obs:static-baseline}, the regret of any algorithm for static databases depends only on the support size $\k$, with a trivial baseline of $\k-1$. In this section, we show sharper regret bounds under both squared and absolute loss functions. For squared loss, we propose the Sequential Maximum Entropy algorithm (\maxent) in \Cref{subsec:static_squared_upper_bound}, achieving $O(\log \k)$ regret bound. The matching lower bound of $\Omega(\log \k)$ for range queries in \Cref{subsec:static_squared_lower_range} immediately extends to subset and linear queries. For absolute loss, the regret is tight at $\tilde{\Theta}(\sqrt{k})$: we prove a lower bound of $\Omega(\sqrt{k/\log k})$ for subset queries in \Cref{subsec:static_absolute_lower_bound}, and an upper bound of $O(\sqrt{k \log k})$ for general linear queries in \Cref{subsec:static_absolute_upper_bound}, achieved by using the same \maxent\ algorithm. For the special case of point queries, the regret is constant under both squared and absolute loss. We defer the detailed analysis to Appendix~\ref{app:static-point}.
\subsection{Upper Bound under Squared Loss: The Sequential Maximum Entropy Algorithm}\label{subsec:static_squared_upper_bound}
As presented in \Cref{alg:seq_max_ent}, the {\em Sequential Maximum Entropy (\maxent)}
algorithm maintains a feasible region $\mathcal C_\t$ and a
database estimate $\predw_\t$. $\mathcal{C}_\t$ is the set of database vectors consistent with all query-selectivity equalities observed before round $\t$, and $\predw_{\t}$ is the maximum-entropy distribution in $\mathcal{C}_\t$ (intuitively, the ``most uniform'' distribution consistent with the observations).
\cut{We now present the {\em Sequential Maximum Entropy (\maxent)} algorithm that achieves an $O(\log \k)$ regret bound for linear queries, matching the lower bound up to a constant factor. The \maxent\ algorithm (\Cref{alg:seq_max_ent}) maintains a feasible region $\mathcal{C}_\t$ and a database estimate $\predw_{\t}$: $\mathcal{C}_\t$ is the set of database vectors consistent with all query-selectivity equalities observed before round $\t$, and $\predw_{\t}$ is the 
maximum-entropy distribution in $\mathcal{C}_\t$ (intuitively, the ``most uniform'' distribution consistent with the observations).}
% a database estimate $\predw_{\t}$, which is ``as uniform as possible'' (maximizing entropy) while satisfying all observed query-selectivity constraints.
In Line~\ref{line:initialize}, the algorithm initializes $\mathcal{C}_1 = \Delta_\k$ and $\predw_1 = (1/\k, \dots, 1/\k)$, the maximum-entropy distribution over $\mathcal{C}_1$. In each round $\t$, upon receiving query $\queryvec_{\t}$, it predicts $\predselect_{\t} = \langle \queryvec_\t, \predw_\t\rangle$. After observing $\trueselect_\t$, it refines the feasible region to $\mathcal{C}_{\t+1} = \mathcal{C}_\t \cap \{\w \in \Delta_\k : \langle \queryvec_\t, \w\rangle = \trueselect_\t\}$ and updates $\predw_{\t+1}$ to the maximum-entropy distribution over $\mathcal{C}_{\t+1}$.

The update in Line~\ref{line:update_weight} is a standard convex program, for which an $\varepsilon$-accurate solution can be computed in time polynomial in $\k$, $\t$, and $\log(1/\varepsilon)$ by standard convex-optimization methods (e.g., interior-point methods)~\cite{boyd2004convex}. The regret analysis below is stated for the exact optimizer.

\begin{algorithm}[!htbp]
    \caption{Sequential Maximum Entropy (\maxent)}
    \label{alg:seq_max_ent}
    \begin{algorithmic}[1]
        %\label{line:require}$\Require$ $k$
        \State \textbf{Initialize:} Feasible region $\mathcal{C}_1 \gets \Delta_k$, $\predw_1 \gets (1/k, \dots, 1/k)$ \label{line:initialize}
        \For{$t = 1, 2, \dots, T$} \label{line:forloop}
            \State \textbf{Predict:} Receive query $\queryvec_{\t}$ and make a prediction on the selectivity $\predselect_{\t} \gets \langle \queryvec_{\t}, \predw_{\t} \rangle$ \label{line:predict}
            \State \textbf{Observe:} Receive true selectivity $\sigma_t$ \label{line:observe}
            \State \textbf{Update Feasible Region:} $\mathcal{C}_{\t+1} \gets \mathcal{C}_{\t} \cap \left\{ \predw \in \Delta_k \mid \langle \mathbf{v}_t, \predw \rangle = \sigma_t \right\} \label{line:constraint_set_eq}$ \label{line:update_constraint}

            \State \textbf{Update Database Estimate:} Compute $\predw_{\t+1}$ by maximizing entropy over $\mathcal{C}_{\t+1}$:\label{line:update_weight}
            \begin{align*}
                \predw_{\t+1} \gets \arg\max_{\w \in \mathcal{C}_{\t+1}} \left( -\sum_{i=1}^k w_i \ln (w_i) \right) \label{line:maxent}
            \end{align*}
            %\hfill \textit{(Equivalently: $\arg\min_{\mathbf{w} \in \mathcal{C}_{\t+1}} D_{\mathrm{KL}}(\mathbf{w} \,\|\, \predw_1)$)}
        \EndFor \label{line:endfor}
    \end{algorithmic}
\end{algorithm}

We now analyze the regret of the \maxent\ algorithm.

\begin{theorem}
    \label{thm:static_upper}
    Under squared loss in the static database setting with general linear queries, the \maxent\ algorithm (\Cref{alg:seq_max_ent}) achieves the following regret bound:
      %The \maxent\ algorithm (\Cref{alg:seq_max_ent}) guarantees the following regret bound for squared loss under the static database setting:
      \begin{align*}
        \regret_{\T} \leq 2 \ln \k = O(\log \k).
      \end{align*}
    \end{theorem}
% \sum_{\t=1}^{\T} \loss_{\t} 

Our analysis uses the \emph{Kullback--Leibler (KL) divergence}: for $\mathbf{p}, \mathbf{q} \in \Delta_\k$,
\begin{align*}
    D_{\text{KL}}(\mathbf{p} \| \mathbf{q}) = \sum_{i=1}^{\k} p_i \ln \frac{p_i}{q_i},
\end{align*}
with the convention $0\ln 0 = 0$. It is non-negative and satisfies $D_{\text{KL}}(\mathbf{p} \| \mathbf{q}) = 0$ iff $\mathbf{p} = \mathbf{q}$~\cite{cover1999elements}. Since $\predw_1$ is the uniform distribution $(1/\k, \dots, 1/\k)$, for any $\w \in \Delta_\k$ we have:
\begin{equation}\label{eq:KL0}
    D_{\text{KL}}(\w \| \predw_1) = \sum_{i=1}^{\k} w_i \ln(\k w_i) = \sum_{i=1}^{\k} w_i \ln w_i + \ln \k \le \ln \k,
\end{equation}
where the second equality uses $\sum_i w_i = 1$, and the inequality uses $w_i \ln w_i \le 0$ for $w_i \in [0,1]$.

We first establish a squared loss bound for each round.
\begin{lemma}
    \label{lemma:single_step_loss}
    For any round $\t \in \{1, 2, \ldots, \T\}$, the squared loss incurred by the \maxent\ algorithm satisfies
    \begin{align*}
        \loss_{\t} \le 2\, D_{\text{KL}}(\predw_{\t+1} \| \predw_{\t}).
    \end{align*}
\end{lemma}
\begin{proof}
    The squared loss in round $\t$ is $\loss_{\t} = (\trueselect_\t - \predselect_{\t})^2 = (\trueselect_\t - \langle \queryvec_{\t}, \predw_\t \rangle)^2$.
    Since the updated database estimate $\predw_{\t+1}$ is in $\mathcal{C}_{\t+1}$, it must satisfy the most recent query-selectivity equality $\langle \queryvec_{\t}, \predw_{\t+1} \rangle = \trueselect_{\t}$. Substituting this into the squared loss:
    \begin{align*}
        \loss_{\t} = \big(\langle \queryvec_{\t}, \predw_{\t+1} \rangle - \langle \queryvec_{\t}, \predw_\t \rangle\big)^2 = \langle \queryvec_{\t}, \predw_{\t+1} - \predw_\t \rangle^2.
    \end{align*}
    Applying H\"older's inequality (Lemma~\ref{lemma:holder_inequality} in Appendix~\ref{app:static_squared_loss}), and noting that $\queryvec_\t \in [0,1]^\k$ implies $\|\queryvec_\t\|_\infty \le 1$,
    \begin{align*}
        \langle \queryvec_{\t}, \predw_{\t+1} - \predw_\t \rangle^2 \le \|\queryvec_{\t}\|_{\infty}^2 \cdot \|\predw_{\t+1} - \predw_\t\|_1^2 \le \|\predw_{\t+1} - \predw_\t\|_1^2.
    \end{align*}
    Finally, Pinsker's inequality (Lemma~\ref{lemma:pinsker_inequality} in Appendix~\ref{app:static_squared_loss}) gives:
    \[
        \loss_{\t} \le \|\predw_{\t+1} - \predw_\t\|_1^2 \le 2\, D_{\text{KL}}(\predw_{\t+1} \| \predw_\t). \qedhere
    \]
\end{proof}

Last, we sum up the one-step loss to complete the proof of \Cref{thm:static_upper}. %Using the above \Cref{lemma:single_step_loss}, we now prove \Cref{thm:static_upper} by telescoping over KL divergences.

\begin{proof}[Proof of \Cref{thm:static_upper}]
    The \maxent\ update in Line~\ref{line:update_weight} is equivalent to the KL projection of the uniform distribution $\predw_1$ onto the current feasible region:
    \begin{align*}
        \predw_{\t}
        = \arg\max_{\w \in \mathcal{C}_{\t}} \left(-\sum_{i=1}^{\k} w_i \ln w_i\right)
        = \arg\min_{\w \in \mathcal{C}_{\t}} D_{\text{KL}}(\w \| \predw_1).
    \end{align*}
    By the Pythagorean theorem for KL divergence (Lemma~\ref{lemma:pythagorean_theorem} in Appendix~\ref{app:static_squared_loss}), for any $\w$ in $\mathcal{C}_{\t}$ we have:
    \begin{align*}
        D_{\text{KL}}(\w \| \predw_1)
        \ge D_{\text{KL}}(\w \| \predw_\t) + D_{\text{KL}}(\predw_\t \| \predw_1).
    \end{align*}

    The feasible regions are nested, i.e., $\mathcal{C}_{\t+1} \subseteq \mathcal{C}_\t$ for all $\t$. Thus $\predw_{\t+1}$ is also in $\mathcal{C}_\t$. Taking $\w = \predw_{\t+1}$ in the inequality above and rearranging yields:
    \begin{align*}
        D_{\text{KL}}(\predw_{\t+1} \| \predw_\t)
        \;\le\; D_{\text{KL}}(\predw_{\t+1} \| \predw_1) \;-\; D_{\text{KL}}(\predw_\t \| \predw_1).
    \end{align*}
    Summing over $\t = 1, \ldots, \T$ telescopes to
    \begin{align*}
        \sum_{\t=1}^{\T} D_{\text{KL}}(\predw_{\t+1} \,\|\, \predw_\t)
        \le D_{\text{KL}}(\predw_{\T+1} \| \predw_1) - D_{\text{KL}}(\predw_1 \| \predw_1)
        \le \ln \k,
    \end{align*}
    where the last inequality uses (\ref{eq:KL0}) and $D_{\text{KL}}(\predw_1 \,\|\, \predw_1) = 0$.
    Combining this with \Cref{lemma:single_step_loss},
    \[
        \regret_{\T} \;=\; \sum_{\t=1}^{\T} \loss_\t
        \;\le\; 2 \sum_{\t=1}^{\T} D_{\text{KL}}(\predw_{\t+1} \,\|\, \predw_\t)
        \;\le\; 2\ln \k. \qedhere
    \]
\end{proof}

\subsection{Lower Bound under Squared Loss: Range Queries}\label{subsec:static_squared_lower_range}

\begin{restatable}{proposition}{binarylowerbound}\label{thm:static_lower}
      For static databases under the squared loss function, for range queries, any (possibly randomized) algorithm $\mathcal{A}$ incurs expected regret of $\Omega(\log k)$ in the worst case, where $k$ is the support size of the domain.
\end{restatable}

\begin{proof}[Proof sketch]
By Yao's minimax principle~\cite{yao1977probabilistic}, it suffices to exhibit a distribution over input sequences on which every \emph{deterministic} algorithm incurs expected regret $\Omega(\log \k)$.
We construct a randomized adversary that puts all the database mass on a uniformly random coordinate $i^*$. The adversary maintains a \emph{candidate set} of indices still consistent with the prior selectivity observations, initialized to $[\k]$. In round $\t$, it issues a range query over one half of the candidate set, and the revealed selectivity restricts the candidate set of the next round to whichever half contains $i^*$. By induction, the posterior on $i^*$ given the algorithm's history remains uniform on the candidate set, so the true selectivity is a fair coin flip for any deterministic algorithm. This forces expected squared loss at least $1/4$ for $\lfloor \log_2 \k \rfloor$ rounds before $i^*$ is identified, yielding $\Omega(\log \k)$ expected regret for any deterministic algorithm by \Cref{obs:static-baseline}. (See Appendix~\ref{app:static_lower} for the detailed proof.)
%and by Yao's minimax principle~\cite{yao1977probabilistic} the same bound holds for the worst-case expected regret of every randomized algorithm. See \Cref{proof:static_lower} for the full argument.
\end{proof}

This matches the $O(\log \k)$ upper bound of \Cref{thm:static_upper}, giving $\Theta(\log \k)$ regret for range, subset, and linear queries.

\subsection{Lower Bound under Absolute Loss: Subset Queries}\label{subsec:static_absolute_lower_bound}

\begin{theorem} \label{thm:static_absolute_lower_bound}
For static databases under the absolute loss function, any (possibly randomized) online learning algorithm incurs $\Omega\left(\sqrt{\frac{\k}{\log \k}}\right)$ expected regret in the worst case, even when queries are restricted to subset queries.
\end{theorem}

To prove this theorem, we start by defining Hadamard matrices.
%\paratitle{Hadamard Matrices.}
\begin{definition}[Normalized Hadamard Matrix~\cite{horadam2012hadamard}]\label{def:hadamard}
  A square matrix $\Hmat \in \{-1, +1\}^{d \times d}$ is a \emph{normalized Hadamard matrix} 
  of order $d$ if (i) the rows of $\Hmat$ are mutually orthogonal, and 
  (ii) the first row and first column of $\Hmat$ consist entirely of $+1$s.
  \end{definition}
Denote the $i$-th row of $\Hmat$ by $\hvec_i^\top$. Condition~(i) gives that $\langle \hvec_i, \hvec_j \rangle = 0$ for all $j \neq i$, and $\langle \hvec_i, \hvec_i \rangle = d$; together with~(ii), this implies $\langle \ones, \hvec_i \rangle = 0$ for every $i \geq 2$. It is well known~\cite{sylvester1867lx} that a normalized Hadamard matrix of order $d$ always exists when $d$ is a power of $2$.

\smallskip
\paratitle{Construction.}
Let $d$ be the largest power of $2$ at most $k$. The construction uses only the first $d$ coordinates, and the last $k-d$ coordinates carry zero weight.

First, we construct the candidate database vector using a normalized Hadamard matrix $\Hmat$ of order $d$, and a random vector $\Yvec = (Y_1, Y_2, \dots, Y_d)^\top$ with $Y_1 = 1$ fixed and $Y_2, \dots, Y_d$ independently and uniformly drawn from $\{-\delta, +\delta\}$. $\delta > 0$ is a parameter to be chosen later. Define the candidate database vector $\tildew \in \mathbb{R}^k$ as
\begin{equation}\label{eq:unconstrained_weight_vector}
\tildew = \begin{pmatrix} \frac{1}{d}\, \Hmat^\top \Yvec \\ \mathbf{0}\end{pmatrix},
\end{equation}
where $\mathbf{0}$ denotes the all-zero vector of dimension $k - d$.

Next, we define the query sequence. For $t = 1$, we use an all-ones vector: $\queryvec_1 = \ones \in \{0, 1\}^k$. Since the absolute loss is non-negative, round $1$'s contribution to the regret is non-negative, so we obtain a valid lower bound by considering rounds $t = 2, \ldots, d$ only.
For $t = 2, \ldots, d$, the queries are constructed from the rows of $\Hmat$:
\begin{equation}\label{eq:lb_query_construction}
  \queryvec_t = \begin{pmatrix} \frac{1}{2}(\ones + \hvec_t) \\ \mathbf{0} \end{pmatrix} \in \{0, 1\}^k.
\end{equation}
Since $\hvec_t \in \{-1, +1\}^d$, the first $d$ coordinates of $\queryvec_t$ are either $0$ or $1$, 
and the remaining $k - d$ coordinates are also $0$; hence $\queryvec_t$ is a valid subset query. 
The properties of this construction are captured by the three lemmas below.

\smallskip
\paratitle{Proof of \Cref{thm:static_absolute_lower_bound}.}
We first establish three lemmas.
\begin{lemma}\label{lem:hadamard_extraction}
For each round $t \in \{2, 3, \ldots, d\}$, the true selectivity of query $\queryvec_t$ defined in~\eqref{eq:lb_query_construction} for $\tildew$ defined in~\eqref{eq:unconstrained_weight_vector} is $\trueselect_t = \frac{1}{2}(1 + Y_t)$.
\end{lemma}

\begin{proof}
The true selectivity is $\trueselect_{\t} = \langle \queryvec_{\t}, \tildew \rangle$. Since both $\queryvec_t$ and $\tildew$ have zero in their last $k - d$ coordinates, the inner product reduces to the first $d$ coordinates:
\begin{equation}\label{eq:hadamard_extraction}
\begin{aligned}
    \trueselect_t
    &= \left\langle \frac{1}{2}(\ones+\hvec_t),\; \frac{1}{d}\Hmat^\top \Yvec \right\rangle
    = \frac{1}{2d} (\ones + \hvec_t)^\top \Hmat^\top \Yvec
    = \frac{1}{2d} \left( \ones^\top \Hmat^\top + \hvec_t^\top \Hmat^\top \right) \Yvec.
\end{aligned}
\end{equation}
    %Recall that $\ones^\top = \hvec_1^\top$.
  By the orthogonality of normalized Hadamard matrices, we have $\ones^\top \Hmat^\top = \hvec_1^\top \Hmat^\top = d \evec_1^\top$ and $\hvec_t^\top \Hmat^\top = d \evec_t^\top$, where $\evec_1, \evec_t \in \{0,1\}^d$ denote the first and $t$-th standard basis vectors.
  Substituting these results back into~\eqref{eq:hadamard_extraction}, and recalling that $Y_1 = 1$:
    \[
      \trueselect_t
      = \frac{1}{2d} (d \evec_1^\top + d \evec_t^\top)\Yvec
      = \frac{1}{2}(Y_1 + Y_t)
      = \frac{1}{2}(1 + Y_t). \qedhere
    \]
    %Therefore, the true selectivity at each round depends only on the single random variable $Y_t$.
    \end{proof}
Therefore, the true selectivity in each round depends only on a single random variable. The following lemma bounds the expected loss in each round.

\begin{lemma}\label{lem:randomness_loss}
  For each round $t \in \{2, \ldots, d\}$, let $\predselect_t \in [0, 1]$ be the deterministic algorithm's prediction given $Y_1, \ldots, Y_{t-1}$. Let $\loss_t(\tildew) = |\predselect_t - \langle \queryvec_t, \tildew \rangle|$ be the absolute loss at round $t$ for $\tildew$. Then the expected loss over $Y_t$ for $\tildew$ is given by:
  \begin{align*}
    \E_{Y_t}[\loss_t (\tildew) \mid Y_1, \ldots, Y_{t-1}] \ge \frac{\delta}{2}.
  \end{align*}
\end{lemma}

\begin{proof}
Fix $(Y_1, \ldots, Y_{t-1})$.
By \Cref{lem:hadamard_extraction}, $\trueselect_t = \frac{1}{2}(1 + Y_t)$, thus the absolute loss is:
\begin{align*}
    \loss_t (\tildew) = |\predselect_t - \trueselect_t| = \left| \predselect_t - \frac{1}{2}(Y_t + 1) \right| 
    = \frac{1}{2} |(2\predselect_t - 1) - Y_t|.
\end{align*}
Define $\hat{y}_t \coloneqq 2\predselect_t - 1 \in [-1, 1]$, so that 
$\loss_t (\tildew) = \frac{1}{2}|\hat{y}_t - Y_t|$.

Since the algorithm is deterministic, $\hat{y}_t$ is a deterministic function of $(Y_1, \ldots, Y_{t-1})$. $Y_t = \pm \delta$ with equal probability and is independent of $(Y_1, \ldots, Y_{t-1})$. Therefore,
\begin{equation}\label{eq:expect_single_loss_lower_bound}
    \E_{Y_t}[|\hat{y}_t - Y_t| \mid Y_1, \ldots, Y_{t-1}] 
    = \frac{1}{2}| \delta - \hat{y}_t| + \frac{1}{2}|\hat{y}_t + \delta| 
    \ge \delta,
\end{equation}
where the inequality follows from the triangle inequality 
$|a| + |b| \ge |a + b|$ applied with $a = \delta - \hat{y}_t$ and 
$b = \hat{y}_t + \delta$. Therefore, 
$\E_{Y_t}[\loss_t(\tildew) \mid Y_1, \ldots, Y_{t-1}] = \frac{1}{2}\E_{Y_t}[|\hat{y}_t - Y_t| \mid Y_1, \ldots, Y_{t-1}] \ge \delta/2$.
\end{proof}

%The constructed $\tildew$ may violate the non-negativity constraint of $\Delta_k$ for extreme realizations of $Y_2, \dots, Y_k$.

So far, we have shown that our construction of $\tildew$ and $\queryvec_t$ forces any deterministic algorithm to incur at least $\delta/2$ expected absolute loss in each $\t \in \{2, \ldots, d\}$. The next lemma shows that the probability of $\tildew$ being a valid database state depends on $\delta$ and $d$. This implies that by choosing $\delta$ sufficiently small, the probability can be made arbitrarily close to $1$.

\begin{lemma}\label{lem:database_validity}
  Let $E = \{\tildew \in \Delta_k\}$ denote the event that the candidate weight vector $\tildew$ defined in~\eqref{eq:unconstrained_weight_vector} lies in the simplex $\Delta_k$. The probability that $E$ holds is $\mathbb{P}(E) \ge 1 - 2d\exp\!\left(-\frac{1}{2d\delta^2}\right)$.
  %Particularly, by setting $\delta = \frac{1}{\sqrt{10 k \ln \k}}$, the probability of $E$ is at least $1- \frac{1}{4} \delta$.
  %By setting $\delta = 1/\sqrt{2k \ln(8k)}$, the probability of $E$ is at least $3/4$.
  % In particular, for $\delta = 1/\sqrt{Ck\ln k}$ with any sufficiently large constant 
  % $C > 0$, we have $\mathbb{P}_P(\mathbf{w} \in \Delta_k) \geq 1 - 1/k^3$.
\end{lemma}

\begin{proof}
For $\tildew$ to be in the simplex $\Delta_k$, it must satisfy: (1) $\sum_{i=1}^k \tilde{w}_i = 1$ and (2) $\tilde{w}_i \geq 0$ for all $i \in [k]$. We first show that the first constraint holds for every realization of $\Yvec$: %because $\hvec_1^\top \Hmat^\top= d \evec_1^\top$ and $Y_1 = 1$:
\begin{align*}
  \sum_{i=1}^k \tilde{w}_i = \sum_{i=1}^d \tilde{w}_i = \frac{1}{d} \ones^\top \Hmat^\top \Yvec = \frac{1}{d} \cdot d\, \evec_1^\top \Yvec = Y_1 = 1.
\end{align*}

For the second constraint, the last $k - d$ coordinates of $\tildew$ are zero, so we focus on the first $d$ coordinates. Each $\tilde{w}_i$ for $i \in \{1, \ldots, d\}$ is of the form $\tilde{w}_i = \frac{1}{d} + \frac{1}{d} \sum_{s=2}^d Y_s h_{s,i}$, where $h_{s,i}$ denotes the $(s,i)$-entry of $\Hmat$. Thus the non-negativity condition $\tilde{w}_i \geq 0$ is equivalent to $\sum_{s=2}^{d} Y_s h_{s,i} \geq -1$. Let $\tilde{E}$ denote the event that $|\sum_{s=2}^d Y_{s} h_{s,i}| \leq 1$ for all $i \in \{1, \ldots, d\}$, i.e.,
\begin{align*}
\tilde{E} = \bigcap_{i=1}^d \left\{ \left|\sum_{s=2}^d Y_{s} h_{s,i}\right| \le 1 \right\}.
\end{align*}

If $\tilde{E}$ holds, then $\sum_{s=2}^d Y_s h_{s,i} \geq -1$ for all $i \in \{1, \ldots, d\}$. Hence $\tilde{E} \subseteq
E$, and it suffices to lower bound $\mathbb{P}(\tilde{E})$.
We first bound the probability of the complement event, denoted by $\tilde{E}^c$:
\begin{align*}
\tilde{E}^c = \bigcup_{i=1}^d \left\{ \left|\sum_{s=2}^d Y_{s} h_{s,i}\right| > 1 \right\}.
\end{align*}

%Under the unconditioned distribution $P$,
Fix a coordinate $i$. The terms $Y_s h_{s,i}$ for $s \in \{2, \ldots, d\}$ are independent random variables taking values in $\{-\delta, +\delta\}$ with equal probability. Thus the expectation of the sum is zero. %: $\mathbb{E}[\sum_{s=2}^d Y_{s} h_{s,i}] = 0$. 
Applying Hoeffding's inequality\footnote{We use the standard two-sided Hoeffding inequality~\cite{hoeffding1963}: for independent random variables $X_1, \ldots, X_n$ with $X_i \in [a_i, b_i]$ almost surely, $\mathbb{P}\!\left(\left|\sum_{i=1}^n X_i - \E\!\left[\sum_{i=1}^n X_i\right]\right| \geq \epsilon\right) \leq 2\exp\!\left(-\frac{2\epsilon^2}{\sum_{i=1}^n (b_i - a_i)^2}\right)$ for any $\epsilon > 0$. The same bound applies to $\mathbb{P}(|\cdot| > \epsilon)$, since $\{|\cdot| > \epsilon\} \subseteq \{|\cdot| \geq \epsilon\}$.} gives:
\begin{align*}
\mathbb{P}\left( \left| \sum_{s=2}^d Y_{s} h_{s,i} \right| > 1 \right) \le 2 \exp\left( - \frac{2 \cdot 1^2}{(d-1) (2\delta)^2} \right) \le 2 \exp\left(- \frac{1}{2 d \delta^2} \right).
\end{align*}
A union bound over the $d$ coordinates gives $\mathbb{P}(\tilde{E}^c) \leq 2d\exp\!\left(-\frac{1}{2d\delta^2}\right)$, which implies that:
\[
  \mathbb{P}(E) \geq \mathbb{P}(\tilde{E}) = 1 - \mathbb{P}(\tilde{E}^c) \ge 1 - 2d\exp\left(-\frac{1}{2d\delta^2}\right). \qedhere
\]

  \end{proof}
  Combining above results, we can now complete the proof.

\begin{proof}[Proof of \Cref{thm:static_absolute_lower_bound}]\label{proof:static_absolute_lower_bound}
%By Yao's minimax principle, it suffices to show that the randomized adversary described above forces every deterministic algorithm to incur $\Omega(\sqrt{k/\log k})$ expected cumulative loss.

The adversary constructs the true database vector $\wstar \in \Delta_k$ as follows: it draws $\Yvec$ and computes the candidate database vector $\tildew$ via~\eqref{eq:unconstrained_weight_vector}. If $\tildew \in \Delta_k$, i.e., $\tildew$ is a valid database vector, it sets $\wstar = \tildew$; otherwise, it sets $\wstar = \frac{1}{k}\ones$. This guarantees that $\wstar$ is always valid. Recall that we denote the event $E=\{\tildew \in \Delta_k\}$, and denote its complement by $E^c$.

The query sequence is given by~\eqref{eq:lb_query_construction}. Let $\loss_t(\wstar) = |\predselect_t - \langle \queryvec_t, \wstar \rangle|$ denote the algorithm's absolute loss at round $t$ for $\wstar$. Fix a deterministic algorithm and a round $t \in \{2, \ldots, d\}$. The law of total expectation gives that the expected loss for $\wstar$ over $\Yvec$ is:
\begin{align*} %\label{eq:lb_total_exp}
  \E_{\Yvec}[\loss_t(\wstar)] = \mathbb{P}(E) \cdot \E_{\Yvec}[\loss_t(\wstar) \mid E] + \mathbb{P}(E^c) \cdot \E_{\Yvec}[\loss_t(\wstar) \mid E^c].
\end{align*}

The first term corresponds to the case where $\wstar=\tildew$, and the second corresponds to the case where the adversary outputs the uniform fallback ($\wstar= \frac{1}{k} \ones$). Since the absolute loss is non-negative, we can drop the second term. This yields:
%Furthermore, under event $E$, the true database is exactly the unconstrained vector ($\wstar = \w$).
\begin{equation}\label{eq:lb_event_e_bound}
  \E_\Yvec[\loss_t(\wstar)] \ge \mathbb{P}(E) \cdot \E_{\Yvec}[\loss_t(\tildew) \mid E].
\end{equation}

To lower-bound the right-hand side, we apply the law of total expectation to the loss for the candidate database vector $\tildew$:
\begin{align*}
  \E_\Yvec[\loss_t(\tildew)] = \mathbb{P}(E) \cdot \E_{\Yvec}[\loss_t(\tildew) \mid E] + \mathbb{P}(E^c) \cdot \E_{\Yvec}[\loss_t(\tildew) \mid E^c].
\end{align*}
Rearranging this equation isolates the exact term we need:
\begin{equation}\label{eq:lb_rearranged}
  \mathbb{P}(E) \cdot \E_{\Yvec}[\loss_t(\tildew) \mid E] = \E_\Yvec[\loss_t(\tildew)] - \mathbb{P}(E^c) \cdot \E_{\Yvec}[\loss_t(\tildew) \mid E^c].
\end{equation}

By \Cref{lem:randomness_loss} and the tower property, we have: 
\begin{align*}
  \E_\Yvec[\loss_t(\tildew)] = \E_{Y_1, \ldots, Y_{t-1}} \left[\E_{Y_t}[\loss_t(\tildew) \mid Y_1, \ldots, Y_{t-1}]\right] \ge \delta/2.
\end{align*}
By \Cref{lem:hadamard_extraction}, $\langle \queryvec_t, \tildew \rangle = \frac{1}{2}(1+Y_t) \in [0,1]$, and $\predselect_t \in [0,1]$, so the absolute loss satisfies $\loss_t(\tildew) \le 1$ deterministically; thus $\E_{\Yvec}[\loss_t(\tildew) \mid E^c] \le 1$. %Substituting into \eqref{eq:lb_rearranged}:
% \begin{equation}\label{eq:lb_substituted}
%   \mathbb{P}(E)\,\E[\loss_t(\tildew) \mid E] \ge \frac{\delta}{2} - \mathbb{P}(E^c).
% \end{equation}
By \Cref{lem:database_validity}, we know that $\mathbb{P}(E^c) \le 2d \exp(-\frac{1}{2d\delta^2})$. Setting $\delta = \frac{1}{4\sqrt{d \ln d}}$ yields:
\begin{align*}
  \mathbb{P}(E^c) \le 2d \exp\left(- \frac{16 d \ln d}{2 d} \right) = 2d^{-7}.
  \end{align*}

For any integer $d \ge 2$, it is straightforward to verify that $2d^{-7} \le \frac{1}{16\sqrt{d \ln d}} = \frac{\delta}{4}$.
Substituting these bounds into~\eqref{eq:lb_rearranged} gives:
\begin{align*}
\mathbb{P}(E) \cdot \E_\Yvec[\loss_t(\tildew) \mid E] \ge \frac{\delta}{2} - \left(\frac{\delta}{4} \cdot 1\right) = \frac{\delta}{4}.
\end{align*}
Combining with~\eqref{eq:lb_event_e_bound}, we conclude that the expected loss for the true database state at every single round is at least $\delta/4$, i.e., $\E_\Yvec[\loss_t(\wstar)] \ge \delta/4$.

By \Cref{obs:static-baseline}, the algorithm's expected regret is exactly its expected cumulative loss. Summing over the $d - 1$ query rounds indexed $t = 2, \ldots, d$ we can lower bound the regret as:
\[
\E[\regret_T] 
= \sum_{t=2}^{d} \E_\Yvec[\loss_t(\wstar)] \ge (d-1) \cdot \frac{\delta}{4} = \frac{d-1}{16\sqrt{d \ln d}}
%& = \Omega\!\left(\sqrt{d/\log d}\right) 
= \Omega\!\left(\sqrt{k/\log k}\right).
\]
The last step follows from the fact that $d$ is the largest power of $2$ at most $k$, thus $d > k/2$. 

Finally, by Yao's minimax principle, we conclude that any randomized algorithm also incurs $\Omega(\sqrt{k/\log k})$ expected regret against worst-case inputs.
\end{proof}

\subsection{Upper Bound under Absolute Loss}\label{subsec:static_absolute_upper_bound}
The following is a standard relationship between the squared and absolute loss functions:
\begin{observation}
  \label{obs:squared-to-absolute}
  Let $\{\trueselect_t\}_{t=1}^T$ and $\{\predselect_t\}_{t=1}^T$ be arbitrary sequences in $[0,1]$. If $\sum_{t=1}^T (\trueselect_t - \predselect_t)^2  \le \alpha$,
  then $\sum_{t=1}^T |\trueselect_t - \predselect_t| \le \sqrt{\alpha T}$.
  \end{observation}
Then, the absolute-loss upper bound follows from the squared-loss bound (\Cref{thm:static_upper}) by observing that \maxent\ incurs non-zero loss on at most $\k-1$ rounds. A key property of \maxent\ is that its estimate $\predw_\t$ is consistent with all previously query-selectivity 
observations: $\langle \queryvec_s, \predw_\t \rangle = \trueselect_s$ for all $s < 
\t$. Consequently, whenever a query is a linear combination of previous queries, 
\maxent\ predicts the selectivity exactly and incurs zero loss. Combined with the simplex normalization $\sum_i \w_i = 1$, at most $\k-1$ linearly independent queries are needed to determine $\w^*$ uniquely, so \maxent\ incurs non-zero loss on at most $\k-1$ rounds under either loss function.

Restricting \Cref{obs:squared-to-absolute} to the (at most) $\k-1$ rounds with non-zero loss and using the cumulative squared-loss bound $O(\log \k)$ from \Cref{thm:static_upper}, we obtain a cumulative absolute loss of $O(\sqrt{(\k-1) \log \k}) = O(\sqrt{\k \log \k})$, independent of the time horizon $\T$. Combined with \Cref{obs:static-baseline}, this yields:

\begin{corollary}\label{thm:static_absolute_upper}
The \maxent\ algorithm (\Cref{alg:seq_max_ent}) achieves $O(\sqrt{\k \log \k})$ regret for static databases under absolute loss, with general linear queries.
\end{corollary}

%% file: sections_arxiv/dynamic_squared_loss.tex
\section{Dynamic Databases}\label{sec:dynamic_squared_loss}
We now turn to dynamic databases, where the underlying database state may change from round to round.
As noted in \Cref{subsec:static-dynamic-compare}, the regret in this setting depends inherently on the time horizon $\T$. We therefore treat $\T$ as the dominant parameter and establish nearly tight bounds for both loss functions.

For squared loss, the regret is tight at $\Theta(\log T)$. We give details of the lower bound of $\Omega(\log T)$ in the following \Cref{thm:dynamic_lower} (\Cref{subsec:dynamic_squared_loss_lower_bound}). The upper bound of $O(\k \log T)$ is a direct application of the \textit{Exponentially Weighted Online Optimization} (EWOO) algorithm, originally introduced by ~\cite{hazan2007logarithmic}. We defer the EWOO algorithm and its regret analysis to Appendix~\ref{app:dynamic_squared_loss_upper_bound}.

For absolute loss, we show that the regret grows as $\sqrt{T}$, and further establish a quantitative gap between different query classes: the tight regret bound is $\Theta(\sqrt{T})$ for point queries, while for the more general subset and linear queries the tight bound is $\Theta(\sqrt{T\log k})$. Both regret bounds are achieved by the \textit{Follow-the-Regularized-Leader} (FTRL) algorithm with different regularization functions~\cite{shalev2012online}. We defer the detailed analysis of lower and upper bounds to Appendix~\ref{app:dynamic_absolute_loss}.

\subsection{Lower Bound under Squared Loss}\label{subsec:dynamic_squared_loss_lower_bound}
To establish the lower bound under the squared loss function, we construct a stochastic instance with support size $\k=2$, in which the adversary fixes the query vector to be $(1, 0)$ but randomizes the database weights. Specifically, the adversary generates the true selectivity as a sequence of independent Bernoulli trials with a parameter $\prior$ drawn uniformly from $[0, 1]$. As we will show in the proof, this construction reduces the selectivity estimation problem to that of estimating $\prior$ from sequential observations. We show in the proof below that any deterministic algorithm will incur a regret of $\Omega(\log T)$ for this instance. By Yao's minimax principle, this then implies a lower bound of $\Omega(\log T)$ on the regret of any randomized algorithm.

\begin{theorem}
    \label{thm:dynamic_lower}
    For dynamic databases under the squared loss function, any (randomized) online algorithm $\mathcal{A}$ incurs $\Omega(\log T)$ (expected) regret, even for point queries. 
  \end{theorem}

\begin{proof}    
Let the support size $\k = 2$. The adversary fixes the query vector to the point query $\queryvec_{\t} = (1, 0)$ for all rounds $\t = 1, \ldots, \T$, i.e., the query asks for the weight of the first element. At the beginning (before $t=1$), the adversary draws a parameter $\prior \sim \mathrm{Uniform}[0,1]$. Then, in each round $t$, the adversary generates the database vector $\w_{\t}$ for the database as follows:
\begin{equation}\label{eq:weight_construction}
  \w_{\t} = 
  \begin{cases}
    (1,0), & \text{with prob. } \prior \\
    (0,1), & \text{with prob. } 1-\prior,
  \end{cases}
\end{equation}
Consequently, the true selectivity $\trueselect_{\t} = \inner{\queryvec_{\t}}{\w_{\t}} \in \{0, 1\}$ satisfies:
\begin{equation}\label{eq:bernoulli_selectivity}
\trueselect_{\t} \mid \prior \stackrel{\text{i.i.d.}}{\sim} \mathrm{Bernoulli}(\prior), \quad \text{and thus } \E[\trueselect_{\t} \mid \prior] = \prior.
\end{equation}

Since the query is fixed to $\queryvec_{\t} = (1,0)$ for all $t$, any fixed database vector $\mathbf{u} = (u_1, u_2) \in \Delta_2$ yields a constant prediction $\inner{\queryvec_{\t}}{\mathbf{u}} = u_1$ in every round. Thus, comparing against the best fixed database vector (recall the definition of regret in \eqref{eq:regret-static}) is equivalent to comparing against the best constant prediction. Furthermore, since $\prior$ is a valid prediction, the loss of the best fixed solution is upper bounded by $\sum_{t=1}^T (\trueselect_t - \prior)^2$. 

Let $\predselect_{\t}$ denote the selectivity returned by the algorithm in round $\t$. Then, the cumulative loss of the algorithm is given by $\sum_{t=1}^{T} (\predselect_\t - \trueselect_\t)^2$. Therefore, the expected regret of the algorithm is at least:
\begin{align*}
  \E \left[ \sum_{t=1}^{T} (\predselect_t - \trueselect_t)^2 - \min_{x \in [0,1]} \sum_{t=1}^{T} (\trueselect_t - x)^2 \right]
  \ge \sum_{t=1}^{T} \E \left[ (\predselect_t - \trueselect_t)^2  - (\trueselect_t - \prior)^2 \right].
\end{align*}

\sloppy
Then, let us fix $t$. Since the algorithm's prediction $\predselect_{\t}$ depends only on $\history_{\t} = \{\queryvec_1, \trueselect_1, \queryvec_2, \trueselect_2, \ldots, \queryvec_{\t-1}, \trueselect_{\t-1}, \queryvec_{\t}\}$,
we have $\predselect_{\t} \mid \history_{\t}, \prior = \predselect_{\t}$. Similarly, $\trueselect_t$ depends only on $\prior$. Hence, 
$\E[\trueselect_t \mid \history_{\t}, \prior] = \prior$. We condition on $\prior$ and $\history_{\t}$, and write:
\begin{align*}
&\E \left[ (\predselect_t - \trueselect_t)^2  - (\trueselect_t - \prior)^2 \mid \history_{\t}, \prior \right]
= (\predselect_t - \prior) \cdot \E \left[ \predselect_t + \prior - 2 \trueselect_t \mid \history_{\t}, \prior \right] \\
=& (\predselect_t - \prior) \cdot (\predselect_t + \prior - 2\cdot \E \left[ \trueselect_t \mid \history_{\t}, \prior \right]) = (\predselect_t - \prior)^2.
\end{align*}
Taking expectation over $\prior$ and $\history_{\t}$, the expected regret for each round $\t$ becomes:
\begin{align*}
  \E\left[ (\trueselect_{\t}-\predselect_{\t})^2 - (\trueselect_{\t} - \prior)^2\right] = \E\left[(\predselect_{\t} - \prior)^2\right].
\end{align*}

We now find the prediction $\predselect_{\t}$ that minimizes $\E\left[(\predselect_{\t} - \prior)^2 \mid \history_{\t}\right]$.
Note that:
\begin{align*} %\label{eq:mse_expand}
  \E\left[(\predselect_{\t} - \prior)^2 \mid \history_{\t}\right] = \predselect_{\t}^2 - 2\predselect_{\t}\E[\prior \mid \history_{\t}] + \E[\prior^2 \mid \history_{\t}].
\end{align*}
Taking the derivative with respect to $\predselect_{\t}$ and setting it to zero:
\begin{align*} %\label{eq:optimal_prediction_deriv}
  \frac{d}{d\predselect_{\t}} \E\left[(\predselect_{\t} - \prior)^2 \mid \history_{\t}\right] = 2\predselect_{\t} - 2\cdot \E[\prior \mid \history_{\t}] = 0,
\end{align*}
and the second derivative is positive.

Thus, the prediction that minimizes the expected regret for each $\t$ is $\predselect_{\t}^* = \E[\prior \mid \history_{\t}]$, and the minimum value of the expected regret is the variance of $\prior \mid \history_{\t}$, denoted $\mathrm{Var}(\prior \mid \history_{\t})$.

Note that we have $\prior \sim \mathrm{Uniform}[0,1]$. By standard properties of the Beta distribution~\cite{johnson2022bayes,berger2013statistical}, the uniform distribution on $[0,1]$ is the $\mathrm{Beta}(1,1)$ distribution. Moreover, after observing $\t-1$ independent outcomes $\trueselect_1,\ldots,\trueselect_{\t-1}$ from $\mathrm{Bernoulli}(\prior)$, the posterior distribution of $\prior$ given the observations follows the Beta distribution again:
\begin{align*} %\label{eq:posterior_beta}
  \prior \mid \history_{\t} \sim \mathrm{Beta}\bigl(1+H_{\t-1},\, \t - H_{\t-1}\bigr),
\end{align*}
where $H_{\t-1} = \sum_{s=1}^{t-1} \sigma_s$, i.e., the number of rounds where the database is $(1, 0)$ among the first $t-1$ rounds.

The variance of $\mathrm{Beta}(\alpha, \beta)$ is $\frac{\alpha\beta}{(\alpha+\beta)^2(\alpha+\beta+1)}$. 
With $\alpha = 1 + H_{\t-1}$ and $\beta = \t - H_{\t-1}$, we obtain:
\begin{align*}
  \mathrm{Var}(\prior \mid \history_{\t}) = \frac{(1 + H_{\t-1})(\t - H_{\t-1})}{(\t+1)^2(\t+2)}.
\end{align*}
Note that $H_{t-1} \sim \mathrm{Binomial}(t-1, \prior)$, where $\prior \sim \mathrm{Uniform}[0, 1]$. Therefore,
\[
    \Pr(H_{\t-1} = \zeta) 
    = \int_{\prior = 0}^1 {t-1\choose \zeta}\cdot \prior^\zeta \cdot (1-\prior)^{t-1-\zeta} d \prior 
    = \frac{1}{\t} \quad \text{ for } \zeta \in \{0, 1, \ldots, \t-1\}.
\]
Therefore,
\[
  \E[\mathrm{Var}(\prior \mid \history_{\t})]
= \frac{1}{(\t+1)^2(\t+2)} \cdot \sum_{\zeta=0}^{\t-1} \frac{(1+\zeta)(\t-\zeta)}{\t} 
= \frac{1}{\t(\t+1)^2(\t+2)} \sum_{\zeta=0}^{\t-1} \left(\t + (\t-1)\zeta - \zeta^2\right).
\]
Since $\sum_{\zeta=0}^{\t-1} \zeta = \frac{(\t-1)\t}{2}$ and $\sum_{\zeta=0}^{\t-1} \zeta^2 = \frac{(\t-1)\t(2\t-1)}{6}$, we get:
\[
  \E[\mathrm{Var}(\prior \mid \history_{\t})] = \frac{1}{6(\t+1)}.
\]
Summing over $t$, we get that the expected regret of any algorithm is at least $\sum_{t=1}^T \frac{1}{6(t+1)} = \Omega(\log T)$.
\end{proof}

\vspace{-2em}

%% file: sections_arxiv/related_work.tex
\section{Related Work}\label{sec:related}
\paragraph{Learning-Based Selectivity Estimation.}
Machine learning for selectivity (cardinality) estimation has been extensively studied~\cite{dutt2019selectivity,yang2020neurocard,li2023alece,park2020quicksel,wang2023ceda,sun_learned_2021,kim2022learned,wang2021ready}, including systems for handling data and query workload drift~\cite{li2022warper,negi2023robust,wu2024modeling,zeng2024price}.
On the theoretical side, \citet{hu2022selectivity} established PAC sample-complexity bounds for range queries under i.i.d.\ assumptions, and \cite{zeighami2024towards,zeighami2024theoretical,wu2024practical} develop generalization guarantees for learned database operations. However, these results are distribution-dependent, whereas our work provides an \emph{online} framework where data and queries may evolve adversarially, with regret bounds that hold uniformly without distributional assumptions.

\paragraph{Online Learning and Online Convex Optimization.}
Online learning is a well-established framework for sequential decision-making under uncertainty (e.g.,~\cite{cesa2006prediction,shalev2012online,hazan2016introduction,hazan2007logarithmic,zinkevich2003online,kivinen1997exponentiated}). Our selectivity estimation problem fits the framework of online convex optimization (OCO): each query induces a convex loss over the probability simplex, even though the underlying database state is never observed by the learner. Unlike generic OCO analyses, our results exploit the specific query structure and simplex geometry to derive problem-dependent upper and lower bounds.

\paragraph{Online Convex Body Chasing.}
Our static database setting is also related to the convex body chasing problem from online algorithms~\cite{FriedmanL93}. In nested convex body chasing, an online algorithm observes nested convex sets, chooses a point from each newly revealed set, and pays movement cost equal to the distance between consecutive chosen points. Our static database setting has a similar structure: each observed query-selectivity pair shrinks the set of candidate database vectors consistent with the history, from the entire $k$-dimensional simplex to the true database. However, the performance measure is different: we measure scalar prediction error based on projections along the query vector, rather than movement cost, so the algorithms and lower bounds from convex body chasing do not directly apply in our setting. Nevertheless, our use of Hadamard matrices in the lower bound construction was inspired by similar techniques for online convex body chasing~\cite{bubeck2020chasing}.

\paragraph{Online Learning in Databases.}
Recent work~\cite{chesetti2026adapt} applies online-algorithmic ideas to design adaptive filters, addressing a different database problem from selectivity estimation. 
From a more abstract learning-theoretic perspective,
\citet{anderson2025logical} study how learnability transfers from base function classes to derived statistical function classes in both PAC and online learning settings, which is motivated partly by query processing~\cite{hu2022selectivity}. Their agnostic and realizable online cases relate to our dynamic and static database settings, respectively; in particular, their preservation results yield a generic, non-constructive regret upper bound for the agnostic case. Our contribution is problem-specific: for both static and dynamic databases, we give upper and matching lower bounds on regret with concrete algorithms across different query classes and loss functions.

%% file: sections_arxiv/s6-conclusion.tex
\section{Conclusions and Future Work}\label{sec:conclusions}
In this paper, we modeled the selectivity estimation problem for linear queries as an online learning problem, where both queries and data can change arbitrarily over time without distributional assumptions. We established regret bounds for both static and dynamic databases, under both squared and absolute loss functions. %, with explicit dependence on the query class, the support size $\k$, and the time horizon $\T$.
%for static databases under squared loss, we proved tight $\Theta(\log k)$ bounds; for dynamic databases, we established $\Theta(\log T)$ bounds under squared loss and $\Theta(\sqrt{T \log k})$ bounds under absolute loss. 
Several interesting directions remain open for future work:
(i) closing the gap between the $\Omega(\log \T)$ lower bound and $O(\k \log \T)$ upper bound for dynamic databases under squared loss; %, where the optimal dependence on $\k$ remains unknown;
(ii) determining the optimal regret for range queries under absolute loss for both static and dynamic databases;
(iii) extending to richer query classes such as join queries or queries with multiple predicates;
and (iv) empirical evaluation on real-world queries and databases to understand the practical implications of these results.

%% file: sections_arxiv/appedix/app_static_point.tex
\section{Static Databases: Point Queries}\label{app:static-point}
%\subsection{\blue{Warm Up: Point Queries}}\label{sec:static-squared-warmup}
Since each point query asks for exactly one coordinate of the fixed database vector, once a coordinate has been observed, its value is revealed. So all later queries to the same coordinate can be predicted with zero loss. This motivates the following Coordinate Memorization (\coordmem) algorithm in \Cref{alg:coord-mem}. The algorithm maintains a database estimate vector $\hat\w$ and initializes it to $\mathbf 0$. In each round $\t$, it predicts the value of the queried coordinate $i$ from the estimate $\hat w_i$, which is $0$ if the coordinate is not yet queried, or the true weight $w_i^*$ if it has been queried before. After observing the true selectivity $\trueselect_\t = w_i^*$, it memorizes the weight of $i$ in the estimate.

\begin{algorithm}[!htbp]
\caption{Coordinate Memorization (\coordmem)}
\label{alg:coord-mem}
\begin{algorithmic}[1]
    \State \textbf{Initialize:} $\hat\w \gets \mathbf{0}$
            \Comment{Initialize the database estimate vector to the zero vector}
    \For{$\t = 1, 2, \ldots, \T$}
    \State \textbf{Receive} point query $\queryvec_\t = \evec_i$ \Comment{the $i$-th standard basis vector}
    \State \textbf{Predict} $\predselect_\t \gets \hat w_i$
    \State \textbf{Observe} true selectivity $\trueselect_\t = w_i^*$ \Comment{$\wstar$ is the true database vector}
    \State \textbf{Update} $\hat w_i \gets \trueselect_\t$ \Comment{memorize coordinate $i$ with its true value}
    \EndFor
\end{algorithmic}
\end{algorithm}

The following proposition shows that \coordmem\ yields the tight constant-regret bounds under both absolute and squared loss.

\begin{proposition}[Point queries in static databases]\label{prop:static-point}
In the static database setting, for point queries, \coordmem\ achieves regret at most $1$ under both absolute and squared loss. Every online algorithm incurs regret at least $\tfrac{1}{2}$ under absolute loss and at least $\tfrac{1}{4}$ under squared loss in the worst case.
\end{proposition}

\begin{proof}
By Observation~\ref{obs:static-baseline}, in the static setting the regret of any algorithm equals its cumulative loss, so it suffices to bound the cumulative loss of the algorithm.

\smallskip
\noindent\emph{Upper bound.}
\coordmem\ incurs loss only at the first time a coordinate is queried, where it predicts $0$ while the true selectivity is $w^*_i$, incurring absolute loss $w^*_i$ and squared loss $(w^*_i)^2$. Every subsequent query to coordinate $i$ is predicted exactly and incurs zero loss. Summing over all coordinates that are ever queried, we obtain the following upper bounds on the cumulative absolute and squared loss:
\[
    \sum_{\t=1}^{\T} |\trueselect_\t - \predselect_\t| 
    \;\leq\; \sum_{i=1}^{\k} w^*_i 
    \;=\; 1,
    \qquad
    \sum_{\t=1}^{\T} (\trueselect_\t - \predselect_\t)^2 
    \;\leq\; \sum_{i=1}^{\k} (w^*_i)^2 
    \;\leq\; \sum_{i=1}^{\k} w^*_i 
    \;=\; 1,
\]
%where the last inequality uses $w^*_i \in [0,1]$.
Therefore, the regret upper bound is $1$ under both absolute and squared loss.

\smallskip
\noindent\emph{Lower bound.}
Take $\k = 2$, and the adversary queries $\queryvec_1 = (1,0)$ at round $1$. If the algorithm outputs a predicted selectivity $\predselect_1 < 1/2$, the adversary sets $\wstar = (1,0)$ and thus the true selectivity $\trueselect_1 = 1$; otherwise, the adversary sets $\wstar = (0,1)$ and thus $\trueselect_1 = 0$. In either case, the absolute loss is at least $1/2$ and the squared loss is at least $1/4$.
Since there are only 2 coordinates, the algorithm identifies the true database vector after round $1$ and will incur zero loss on all subsequent queries. Therefore, the regret, i.e., the cumulative loss for the static databases, is at least $1/2$ under absolute loss and at least $1/4$ under squared loss.
\end{proof}

%% file: sections_arxiv/appedix/app_static_squared.tex
\section{Details for Section~\ref{sec:static}}
\label{app:static_squared_loss}

\subsection{Deferred Material for \Cref{thm:static_upper}}\label{app:static_squared_upper_loss}
We state the standard information-theoretic lemmas used in the proof of \Cref{thm:static_upper}.

\begin{lemma}[Pinsker's Inequality~\cite{csiszar2011information}]
    \label{lemma:pinsker_inequality}
    For any two probability distributions $\mathbf{p}, \mathbf{q}$ over a finite domain, the $L_1$ norm of the difference satisfies:
    \begin{equation}
        \|\mathbf{p} - \mathbf{q}\|_1 \le \sqrt{2 D_{\mathrm{KL}}(\mathbf{p} \| \mathbf{q})},
    \end{equation}
    where $D_{\mathrm{KL}}(\mathbf{p} \| \mathbf{q}) = \sum_{i=1}^k p_i \ln(\frac{p_i}{q_i})$ is the KL-divergence between $\mathbf{p}$ and $\mathbf{q}$.
\end{lemma}

\begin{lemma}[Pythagorean Theorem for KL-Divergence~\cite{cover1999elements}] 
  \label{lemma:pythagorean_theorem}
  Let $\mathcal{C}$ be a closed convex set and $\mathbf{p}^* = \arg\min_{\mathbf{p} \in \mathcal{C}} D_{\mathrm{KL}}(\mathbf{p} \| \mathbf{q})$ for some distribution $\mathbf{q} \in \mathcal{C}$. Then, for any $\mathbf{p} \in \mathcal{C}$,
  \begin{equation}
      D_{\mathrm{KL}}(\mathbf{p} \| \mathbf{q}) \ge D_{\mathrm{KL}}(\mathbf{p} \| \mathbf{p}^*) + D_{\mathrm{KL}}(\mathbf{p}^* \| \mathbf{q}).
  \end{equation}
\end{lemma}

\begin{lemma}[Hölder's Inequality \cite{SteinShakarchiRealAnalysis}]
  \label{lemma:holder_inequality}
  Let $(\Omega, \mathcal{F}, \mu)$ be a measure space, and let 
$1 < p, q < \infty$ satisfy
\[
\frac{1}{p} + \frac{1}{q} = 1.
\]
If $f$ and $g$ are measurable functions on $\Omega$, then
\[
\int_{\Omega} |f g| \, d\mu
\le 
\left( \int_{\Omega} |f|^{p} \, d\mu \right)^{1/p}
\left( \int_{\Omega} |g|^{q} \, d\mu \right)^{1/q}.
\]
Moreover, if $f$ is bounded, $f \in L^{\infty}(\Omega)$ and $g \in L^{1}(\Omega)$, then
\[
\int_{\Omega} |f g| \, d\mu
\le 
\|f\|_{L^{\infty}(\Omega)} \,
\|g\|_{L^{1}(\Omega)}.
\]
\end{lemma}

\subsection{Lower Bound for Static Databases under Squared Loss: Proof of \texorpdfstring{\Cref{thm:static_lower}}{Theorem~\ref{thm:static_lower}}}
\label{app:static_lower}
We provide the full proof of the following \Cref{thm:static_lower}:
\binarylowerbound*

\begin{proof}
By Yao's minimax principle~\cite{yao1977probabilistic}, to lower-bound the worst-case expected regret of any randomized algorithm, it suffices to construct an input distribution on which every deterministic algorithm incurs expected regret $\Omega(\log k)$.

\noindent\textbf{Adversary construction.}~
Let $r = \lfloor \log_2 \k \rfloor$ and $\k^* = 2^r$, i.e., $\k^*$ is the largest power of 2 less than or equal to $\k$. The adversary uses only the first $\k^*$ coordinates, and the remaining $\k - \k^*$ coordinates carry zero mass and are never queried. Before the game starts, the adversary samples an index $i^*$ from $[\k^*]$ uniformly at random, and sets the static database vector to $\wstar = \evec_{i^*}$. Thus all mass is on a single unknown index.

For each fixed value of $i^*$, the query sequence is deterministic and constructed via a sequence of candidate sets as follows.
For each round $\t \in [r]$, let $S_\t \subseteq [\k^*]$ be the candidate set of indices that could still be the true index $i^*$ given observations before $\t$. Initially, $S_1 = [\k^*]$. Let $S_t^\prime \subset S_t$ be the first half of $S_t$, i.e., the subset containing the $\frac{|S_t|}{2}$ smallest indices.
Then the next candidate set is:
\[
    S_{\t+1} =
    \begin{cases}
        S_t^\prime, & i^* \in S_t^\prime,\\
        S_t \setminus S_t^\prime, & \text{otherwise.}
    \end{cases}
\]
A simple induction shows that each $S_{\t}$ is a contiguous interval of size $2^{r-\t+1}$, and $i^* \mid S_{\t} \sim \mathrm{Uniform}(S_{\t})$: the base case $\t=1$ holds by construction; for the inductive step, assume the claim holds at round $\t$. Then $S_{\t+1}$ is one of the two equal-sized halves of $S_{\t}$ that contains $i^*$, and conditioning a uniform distribution on a subset preserves uniformity.

Then the adversary issues the range query $\queryvec_\t$ such that $v_{t,i} = 1$ if $i \in S_t^\prime$ and $0$ otherwise. The true selectivity is:
\[
    \trueselect_\t =
    \begin{cases}
        1, & i^* \in S_t^\prime,\\
        0, & \text{otherwise.}
    \end{cases}
\]
Since $i^* \mid \{i^* \in S_{\t}\} \sim \mathrm{Uniform}(S_{\t})$ and $S_t^\prime$ contains exactly half of the elements of $S_t$, we conclude $\trueselect_t \mid S_t \sim \mathrm{Bernoulli}(1/2)$.

\noindent\textbf{Per-round expected loss.}~
For any deterministic algorithm $\mathcal{A}$ and a round $\t \in [r]$, its prediction $\predselect_\t$ depends only on the history $\history_{\t} := (\queryvec_1, \trueselect_1, \ldots, \queryvec_{t-1}, \trueselect_{t-1}, \queryvec_t)$. For any $s < t$, each $\trueselect_s$ encodes which half of $S_s$ contains $i^*$; thus the past selectivities $\trueselect_1, \ldots, \trueselect_{t-1}$ uniquely determine $S_\t$. Combined with $i^* \mid S_{\t} \sim \mathrm{Uniform}(S_{\t})$ and $|S_t^\prime| = |S_t|/2$, we have $\trueselect_t \mid \history_{\t} \sim \mathrm{Bernoulli}(1/2)$.

Therefore, the expected loss in each round $\t$ over all possible choices of $i^*$ is:
\[
   \E\!\left[(\trueselect_t - \predselect_t)^2 \mid \history_{\t}\right]
   = \tfrac{1}{2}(1 - \predselect_t)^2 + \tfrac{1}{2}\predselect_t^2
   \;\ge\; \tfrac{1}{4},
\]
where the inequality is minimized at $\predselect_t = 1/2$.

\noindent\textbf{Summing over all rounds.}~
Taking total expectation and summing over $t = 1, \ldots, r$,
\[
   \E\!\left[\sum_{t=1}^{r} (\trueselect_t - \predselect_t)^2\right]
   \;\ge\; \frac{r}{4} \;=\; \frac{\lfloor \log_2 k \rfloor}{4}.
\]
After $r$ rounds, the candidate set $S_{r+1}$ is reduced to a single element, the algorithm has identified the true index $i^*$ and incurs no loss thereafter. Therefore, the expected cumulative squared loss is at least $\tfrac{1}{4}\lfloor \log_2 k \rfloor$, which equals the expected regret by \Cref{obs:static-baseline}. Hence for any deterministic algorithm $\mathcal{A}$,
\[
   \E\!\left[\regret_T(\mathcal{A})\right]
   \;\ge\; \tfrac{1}{4}\lfloor \log_2 k \rfloor = \Omega(\log k).
\]
By Yao's minimax principle, the same lower bound holds for every randomized algorithm against worst-case inputs.
\end{proof}

%% file: sections_arxiv/appedix/app_dynamic_squared.tex
\section{Upper Bound for Dynamic Databases under Squared Loss}\label{app:dynamic_squared_loss_upper_bound}
% \subsection{ Exponentially Weighted Online Optimization}\label{subsec:dynamic_squared_loss_upper_bound}
In this section, we show that the Exponentially Weighted Online Optimization (EWOO) algorithm~\cite{hazan2007logarithmic} achieves an $O(\k \log \T)$ regret bound for dynamic databases under squared loss, by exploiting the exp-concavity of the squared loss function, as shown in the following \Cref{lem:exp_concavity}. The algorithm is given in \Cref{alg:ewoo}, and the regret guarantee is stated in \Cref{cor:ewoo_log_regret}.

In each round $\t$, we view the selectivity prediction $\predselect_\t \in [0,1]$ as 
arising from a predicted database vector $\predw_\t \in \Delta_\k$ via 
$\predselect_\t = \langle \queryvec_\t, \predw_\t \rangle$. The round-$\t$ squared loss, 
as a function of $\predw$, is then
\begin{equation}\label{eq:squared_loss_function}
  \loss_\t(\predw) := \bigl(\trueselect_\t - \langle \queryvec_\t, \predw \rangle\bigr)^2.
\end{equation}
A convex function $f : \mathcal{K} \to \mathbb{R}$ is \emph{$\alpha$-exp-concave} 
over $\mathcal{K}$ if the function $g(\mathbf{x}) := \exp(-\alpha f(\mathbf{x}))$ is concave over 
$\mathcal{K}$~\cite{hazan2016introduction}. The squared loss in \eqref{eq:squared_loss_function} 
satisfies this property, as the following lemma shows.
\begin{lemma}\label{lem:exp_concavity}
  For any $\trueselect_\t \in [0,1]$ and $\queryvec_\t \in [0,1]^\k$, 
  the function $\loss_\t : \Delta_\k \to \mathbb{R}$ defined in 
  \eqref{eq:squared_loss_function} is $\alpha$-exp-concave over $\Delta_\k$ with $\alpha = \tfrac{1}{2}$.
\end{lemma}

\begin{proof}
Let $u(\predw) = \langle \queryvec_{\t}, \predw \rangle$. Because $\predw \in \Delta_\k$ and $\queryvec_{\t} \in [0,1]^\k$, we have $u(\predw) \in [0,1]$. Define the scalar function $g(u) = (\trueselect_{\t} - u)^2$ so that $\loss_{\t}(\predw) = g(u(\predw))$. For twice-differentiable functions, $\alpha$-exp-concavity is implied by the matrix inequality $\nabla^2 \loss_{\t}(\predw) \succeq \alpha (\nabla \loss_{\t}(\predw))(\nabla \loss_{\t}(\predw))^\top$.

By the chain rule, $\nabla \loss_{\t}(\predw) = g'(u) \queryvec_{\t}$ and $\nabla^2 \loss_{\t}(\predw) = g''(u) \queryvec_{\t}\queryvec_{\t}^\top$. The required matrix inequality reduces to the scalar condition $g''(u) \ge \alpha (g'(u))^2$ because both sides are multiples of $\queryvec_{\t}\queryvec_{\t}^\top$. We compute $g'(u) = -2(\trueselect_{\t} - u)$ and $g''(u) = 2$, so $g''(u) \ge \alpha (g'(u))^2$ becomes
\[
2 \ge 4\alpha (\trueselect_{\t} - u)^2.
\]
Since $\trueselect_{\t}, u \in [0,1]$, we have $(\trueselect_{\t} - u)^2 \le 1$, and the inequality holds whenever $\alpha \le \tfrac{1}{2}$. Choosing $\alpha = \tfrac{1}{2}$ satisfies the condition uniformly over $\Delta_\k$, establishing $\alpha$-exp-concavity of $\loss_{\t}$.
\end{proof}

By exploiting this exp-concavity, we achieve logarithmic regret via the Exponentially Weighted Online Optimization (EWOO) algorithm \cite{hazan2007logarithmic}; we follow the description given in \cite{hazan2016introduction}.

As shown in \Cref{alg:ewoo}, the decision space is the probability simplex, i.e., $\mathcal{K} = \Delta_\k$. The parameter $\alpha$ is the exp-concavity parameter of the loss sequence. In each round $\t$, EWOO: (i)~computes a density function $\pi_\t$ over $\Delta_\k$ in Line~\ref{line:compute_density}, which depends only on the cumulative loss up to round $\t-1$; (ii)~predicts the database vector $\predw_\t$ as the centroid (expected value) of $\pi_\t$; (iii)~outputs the selectivity prediction $\predselect_\t = \langle \queryvec_\t, \predw_\t \rangle$; (iv)~observes $\trueselect_\t$ and incurs $\loss_\t(\predw_\t)$. 
Then for the next round $\t+1$, each candidate $\w$ is reweighted by the factor $\exp(-\alpha\loss_\t(\w))$, so candidates with larger squared loss receive smaller weight in future rounds. %Thus the centroid prediction is biased toward database vectors that better explain the past observations.
We restate the general regret bound from \citet{hazan2016introduction} in \Cref{thm:hazan_general}.

\begin{algorithm}[!htbp]
  \caption{Exponentially Weighted Online Optimization (EWOO) for selectivity estimation}
  \label{alg:ewoo}
  \begin{algorithmic}[1]
    \Require Convex set $\mathcal{K} = \Delta_\k$, parameter $\alpha > 0$
    \For{$\t = 1, 2, \ldots, T$}
        \State \textbf{Compute the density function:} $\pi_\t(\w) = \exp \bigl(-\alpha \sum_{\tau=1}^{\t-1} \loss_\tau(\w)\bigr)$ 
               for all $\w \in \Delta_\k$ \label{line:compute_density}
        \State \textbf{Predict the database vector:} 
               $\predw_\t = \dfrac{\int_{\Delta_\k} \w \cdot \pi_\t(\w) \, d\w}
               {\int_{\Delta_\k} \pi_\t(\w) \, d\w}$
               \label{line:predict_weight}
        \State \textbf{Output the predicted selectivity:} 
               $\predselect_\t = \langle \queryvec_\t, \predw_\t \rangle$
               \label{line:output_selectivity}
        \State \textbf{Observe and incur loss:} observe the true selectivity $\trueselect_\t$ and incur 
               $\loss_\t(\predw_\t) = (\trueselect_\t - \langle \queryvec_\t, \predw_\t \rangle)^2$
    \EndFor
  \end{algorithmic}
\end{algorithm}

\begin{theorem}[General EWOO Regret Bound\cite{hazan2016introduction}]
  \label{thm:hazan_general}
  Let the loss functions $\loss_1, \dots, \loss_T$ be $\alpha$-exp-concave over a convex set $\mathcal{K} \subset \mathbb{R}^d$. The regret of the EWOO algorithm is bounded by:
  \begin{align*}
  \regret_T(\text{EWOO}) \le \frac{d}{\alpha} \log T + \frac{2}{\alpha}.
  \end{align*}
  \end{theorem}

  Since the dimension of our problem is $\k$ (the size of the support), and $\alpha = \tfrac{1}{2}$ by \Cref{lem:exp_concavity}, we obtain the following corollary:

\begin{corollary}
\label{cor:ewoo_log_regret}
For dynamic databases under the squared loss function, the EWOO algorithm (\Cref{alg:ewoo}) achieves the following regret bound for the general linear queries:
\begin{align*}
  \label{eq:ewoo_regret_bound}
  \regret_{\T}(\mathrm{EWOO}) \leq 2\k \log T + 4 = O(\k \log T).
\end{align*}
\end{corollary}

Together with the $\Omega(\log T)$ lower bound from \Cref{thm:dynamic_lower}, this establishes the logarithmic regret rate as the optimal regret rate with respect to the time horizon $T$.%, for dynamic databases under the squared loss function. 
%A polynomial gap in $\k$ between the $O(\k \log T)$ upper bound and the $\Omega(\log T)$ lower bound remains open.

\subsubsection*{Polynomial-time algorithm with weaker regret bound: Online Newton Step algorithm}
While EWOO achieves the optimal regret rate, its exact implementation requires integrating over the simplex, which is computationally expensive. Another classical algorithm for exp-concave losses is the Online Newton Step (ONS) algorithm~\cite{hazan2016introduction}, which maintains a single point estimate and updates it using a Hessian approximation. ONS runs in strictly polynomial time, at the cost of a slightly weaker dependence on $\k$ in the regret bound. We restate its general guarantee from \citet{hazan2016introduction} below.

\begin{theorem}[General ONS Regret Bound~\cite{hazan2016introduction}]
  \label{thm:ons_general}
  Let $\loss_1, \dots, \loss_\T$ be $\alpha$-exp-concave loss functions over a 
  convex set $\mathcal{K} \subset \mathbb{R}^d$ with Euclidean diameter $D$, and 
  suppose $\|\nabla \loss_\t(\w)\|_2 \le G$ for all $\w \in \mathcal{K}$ and 
  $\t \in [\T]$. Then for $\T \ge 4$, with parameters 
  $\gamma = \tfrac{1}{2}\min\{\tfrac{1}{GD}, \alpha\}$ and 
  $\varepsilon = \tfrac{1}{\gamma^2 D^2}$, the Online Newton Step algorithm satisfies
  \begin{align*}
  \regret_\T(\mathrm{ONS}) \le 2\left(\frac{1}{\alpha} + GD\right) d \log \T.
  \end{align*}
\end{theorem}

To instantiate this bound for our setting, we note that the simplex $\Delta_\k$ has Euclidean diameter $D = \sqrt{2}$, and the squared-loss gradients satisfy $\|\nabla \loss_\t(\w)\|_2 \le G = 2\sqrt{\k}$ for all $\w \in \Delta_\k$. Plugging $d = \k$, $\alpha = \tfrac{1}{2}$, $D = \sqrt{2}$, and $G = 2\sqrt{\k}$ into \Cref{thm:ons_general} yields
\begin{equation}
  \regret_\T(\mathrm{ONS})
  \;\le\;
  2\bigl(\tfrac{1}{\alpha} + GD\bigr)\,\k \log \T
  = 2\bigl(2 + 2\sqrt{2\k}\bigr)\,\k \log \T
  = O(\k^{3/2}\log \T).
\end{equation}
Thus ONS also achieves logarithmic regret in $\T$, but with a slightly weaker dependence on $\k$ than EWOO. We use EWOO as our main tool because it yields the sharper $O(\k\log \T)$ regret bound.

%% file: sections_arxiv/appedix/app_dynamic_absolute.tex
\section{Dynamic Databases under Absolute Loss}
\label{app:dynamic_absolute_loss}
This section proves the regret bounds under absolute loss stated in \Cref{sec:dynamic_squared_loss}. The lower bounds are established in \Cref{subsec:dynamic_absolute_loss_lower_bound}: \Cref{thm:dynamic_lower_absolute} (\Cref{subsubsec:point_query_lower}) gives an $\Omega(\sqrt{T})$ bound for point queries, and \Cref{thm:dynamic_lower_k_bins} (\Cref{subsubsec:subset_query_lower}) gives an $\Omega(\sqrt{T \log k})$ bound for subset queries. The upper bounds appear in \Cref{subsec:dynamic_absolute_loss_upper_bound}: \Cref{cor:point-query-upper} (\Cref{subsubsec:point_query_upper}) and \Cref{cor:dynamic-abs-upper} (\Cref{subsubsec:linear_query_upper}) instantiate the Follow-the-Regularized-Leader (FTRL) algorithm~\cite{shalev2012online} with squared-$\ell_2$ regularization for point queries and negative-entropy regularization for the general linear queries, respectively. Therefore, the regret for point queries is tight at $\Theta(\sqrt{T})$, independent of the support size $k$.
Since subset queries are a special case of general linear queries, the lower bound for subset queries and the upper bound for linear queries together yield a tight $\Theta(\sqrt{T \log k})$ characterization for both classes.

\subsection{Lower Bound}\label{subsec:dynamic_absolute_loss_lower_bound}
\subsubsection{Point Queries}
\label{subsubsec:point_query_lower}

\begin{theorem}\label{thm:dynamic_lower_absolute}
For dynamic databases under the absolute loss function, for point queries, any (randomized) online algorithm incurs an (expected) regret of $\Omega(\sqrt{\T})$ in the worst case.
\end{theorem}

\begin{proof}
It suffices to consider $k=2$. The adversary fixes the point query $\queryvec_t=(1,0)$ for all $t \in [T]$. Independently in each round, it sets $\w_t=(1,0)$ or $\w_t=(0,1)$ with equal probability. Hence the true selectivity $\trueselect_t=\langle \queryvec_t,\w_t\rangle \sim \text{Bernoulli}(1/2)$, taking values $1$ or $0$ with equal probability.

\sloppy
The algorithm's prediction $\predselect_\t \in [0,1]$ depends only on the history $\history_t = \{\queryvec_1, \trueselect_1, \queryvec_2, \trueselect_2, \ldots, \queryvec_{\t-1}, \trueselect_{\t-1}, \queryvec_{\t}\}$, and is independent of $\w_\t$. Since $\queryvec_\t$ is fixed, this implies that $\predselect_\t$ is independent of $\trueselect_\t = \langle \queryvec_\t, \w_\t\rangle$. Hence the per-round expected loss is
\begin{align*}
\E[\loss_\t] 
= \E\bigl[|\trueselect_\t - \predselect_\t|\bigr] 
= \tfrac{1}{2}(1 - \predselect_\t) + \tfrac{1}{2}\predselect_\t 
= \tfrac{1}{2},
\end{align*}
and by linearity, the expected cumulative loss over $T$ rounds is $\E\bigl[\sum_{\t=1}^\T \loss_\t\bigr] = \tfrac{T}{2}$.

Then we derive an upper bound of the expected total loss of the best fixed predictor. Since the query is fixed, any $\w \in \Delta_2$ yields a constant prediction $c = w_1 \in [0,1]$. Let $H = \sum_{\t=1}^{\T} \trueselect_{\t}$ count the rounds where $\trueselect_{\t} = 1$ (so $\T-H$ counts the rounds where $\trueselect_{\t} = 0$), then the total loss can be written as: 
\begin{align*}
    &\min_{\w \in \Delta_{\k}} \sum_{\t=1}^{\T} \loss_{\t}(\w) = \min_{c \in [0,1]} \sum_{\t=1}^{\T} |\trueselect_{\t} - c|\\
    =& H (1-c) + (\T-H)c = H + (\frac{\T}{2}-H) \cdot 2c
\end{align*}

  % \begin{align*}
  %   \sum_{\t=1}^{\T} |\trueselect_{\t} - c| =
  % \end{align*}
  Choosing $c = 0$ when $H \le T/2$ and $c = 1$ otherwise gives:
  % The majority vote strategy to derive the upper bound of the total loss for the best fixed predictor: if $H \leq T/2$, set $c=0$; if $H > T/2$, set $c=1$. Therfore by Eq. \ref{eq:min_fixed_predictor} and Eq. \ref{eq:majority_vote_loss}:
  \begin{align*}
  \min_{\w \in \Delta_{\k}} \sum_{\t=1}^{\T} \loss_{\t}(\w) \leq \min(H, T-H) = \frac{T}{2} - \frac{|2H-T|}{2}
  \end{align*}
  The last step uses the identity: $\min(a,b) = \frac{a+b}{2} - \frac{|a-b|}{2}, \forall a,b \in \mathbb{R}$.

Therefore, fixing an algorithm, the expected regret is:
\begin{align*}
    \E\left[\sum_{\t=1}^\T \loss_\t - \min_{\w \in \Delta_2} \sum_{\t=1}^\T \loss_\t(\w)\right] \ge \E\left[\frac{T}{2} - \frac{T}{2} + \frac{|2H-T|}{2}\right] = \mathbb{E}\left[|H-\tfrac{T}{2}|\right]
\end{align*}

Since $\trueselect_t \sim \text{Bernoulli}(1/2)$, $H$ follows a binomial distribution $B(T, 1/2)$. By standard properties of the mean absolute deviation of the binomial distribution, we have $\E\left[|H - T/2|\right] \ge \sqrt{\frac{Var(H)}{2}} = \frac{\sqrt{T}}{2\sqrt{2}}$ where $\T \geq 2$~\cite{berend2013sharp}. Therefore, the expected regret of any algorithm against this randomized adversary is $\Omega(\sqrt{T})$.
% \begin{align*}
%     \E\left[\sum_{\t=1}^\T \loss_\t - \min_{\w \in \Delta_2} \sum_{\t=1}^\T \loss_\t(\w)\right] \ge \E\left[\frac{T}{2} - \frac{T}{2} + \frac{|2H-T|}{2}\right] = \mathbb{E}\left[|H-\tfrac{T}{2}|\right] = \Omega(\sqrt{T}).
% \end{align*}
  
% Combining the algorithm's expected loss and the comparator's expected loss, and noting that the supremum over all adversary strategies is at least the expected value under our randomized adversary, we have:
%   \begin{align*}
%     \inf_{\mathcal{A}} \sup \regret_{T}(\mathcal{A}) 
%     &\ge \inf_{\mathcal{A}} \mathbb{E} \left[ \sum_{t=1}^{T} \loss_{t}(\predw_{t}) - \min_{\w \in \Delta_{\k}} \sum_{t=1}^{T} \loss_{t}(\w) \right] \nonumber \\
%     &\ge \frac{T}{2} - \left(\frac{T}{2} - \frac{\sqrt{T}}{2\sqrt{2}}\right) = \Omega(\sqrt{T}).
%   \end{align*}
Moreover, for any $k >2$, the lower bound can be extended by embedding the same construction in the first two coordinates, with zero mass assigned to all remaining coordinates. Since the worst-case regret over adversary strategies is at least the expected regret under any specific adversary distribution, the lower bound holds.
%This concludes the proof.
\end{proof}

\subsubsection{Subset Queries}
\label{subsubsec:subset_query_lower}
% \blue{We now generalize to subset queries over $k \geq 2$ bins. While the point query adversary above uses a fixed query vector, the subset query adversary randomizes both the query and the true selectivity independently, exploiting the richer query structure to force a stronger $\Omega(\sqrt{T\log k})$ lower bound.}
\cut{We now turn to subset queries, where the additional combinatorial
structure of $\queryvec_t \in \{0,1\}^k$ lets the adversary force a
$\log k$ factor that point queries cannot.}
\begin{theorem}\label{thm:dynamic_lower_k_bins}
  For dynamic databases under the absolute loss function, for subset queries, any (possibly randomized) online algorithm incurs an expected regret of $\Omega(\sqrt{T\log k})$ in the worst case.
\end{theorem}

\paratitle{Overview of Techniques.}
In each round $\t \in [T]$, the adversary draws the true selectivity $\trueselect_\t$ and the subset query $\queryvec_\t$ from independent distributions: $\trueselect_\t \sim \mathrm{Bernoulli}(1/2)$, and $\queryvec_\t \in \{0,1\}^\k$ has its first $\k-1$ coordinates drawn i.i.d.\ from $\mathrm{Bernoulli}(1/2)$. The last coordinate is set to $v_{\t,\k} = 1 - v_{\t,1}$ to ensure that $\queryvec_\t$ is never the all-zero or all-one vector. Therefore, there always exists a valid database vector $\w_\t$ that places all its mass on a single coordinate $j \in [\k]$ satisfying $v_{\t,j} = \trueselect_\t$.

Under this distribution, the regret reduces to the expected maximum of $\k-1$ Rademacher random walks of length $\T$. Since $\trueselect_\t$ is independent of the algorithm's prediction, the algorithm incurs expected per-round loss $1/2$, giving total expected loss $T/2$. The best fixed comparator, however, concentrates on one of the $\k$ vertices of $\Delta_\k$, achieving total loss: $\min_{i \in [\k]} \sum_\t |\trueselect_\t - v_{\t,i}|$. Defining $z_{\t,i} := 1 - 2|\trueselect_\t - v_{\t,i}|$, the expected regret becomes $\tfrac{1}{2}\,\E\bigl[\max_{i \in [\k]} \sum_\t z_{\t,i}\bigr]$. By construction, the variables $\{z_{\t,i}\}_{\t \in [T], \, i \in [\k-1]}$ are i.i.d.\ Rademacher random variables taking values in $\{-1, +1\}$ with probability $1/2$; the expected maximum of $\k-1$ such sums of length $T$ is $\Omega(\sqrt{T \log k})$ by standard results in ~\cite{cesa2006prediction} (see \Cref{lem:rademacher_max}).
By Yao's minimax principle, the lower bound on the regret of any randomized algorithm against the worst-case input is $\Omega(\sqrt{T \log k})$.

\begin{proof} %[Proof of \Cref{thm:dynamic_lower_k_bins}]
% The proof proceeds in four steps. First, we construct a randomized adversary that draws random subset queries, and random true selectivities that is independent of the query at each round. Second, the independence forces any algorithm to incur an expected total loss of $T/2$. Third, we show that the best fixed predictor concentrates all weight on whichever coordinate $i \in [k]$ happened to minimize $\sum_{t=1}^T |\trueselect_t - v_{t,i}|$. Intuitively, the per-round loss $|\trueselect_t - v_{t,i}|$ is an indicator of whether coordinate $i$'s query bit disagrees with the true selectivity $\trueselect_t$. The best fixed predictor selects whichever coordinate $i$ accumulated the fewest such disagreements over $T$ rounds. Last, we relate the minimum over these $k$ sums to the maximum of $k$ negatively correlated Rademacher sums. \Cref{lem:neg_corr_walks} shows that the expected maximum deviates above its mean by $\Omega(\sqrt{T\log k})$, which implies that the comparator's expected total loss is at most $T/2 - \Omega(\sqrt{T\log k})$. The regret lower bound follows by taking the difference.

\noindent
\textbf{Step 1: Adversary construction.~} At each round $t \in [T]$, the adversary first draws $k-1$ independent random variables from $\mathrm{Bernoulli}(1/2)$: $b_{t,1}, \ldots, b_{t,k-1}$. Then it defines the subset query $\queryvec_t\in\{0,1\}^k$ by:
\[
  v_{t,i} = b_{t,i} \quad\text{for } i=1,\ldots,k-1,
  \qquad
  v_{t,k} = 1-b_{t,1}.
\]
Since the first and last coordinates always sum to 1, $\queryvec_t$ is never equal to $\zeros$ or $\ones$. Then the adversary draws a true selectivity independent of $\queryvec_t$ and all past variables:
\[
  \trueselect_t \sim \mathrm{Bernoulli}(1/2).
\]
Then it sets the database state $\w_t = \evec_j$, where $j$ is any coordinate satisfying $v_{t,j} = \trueselect_t$. Since $\queryvec_t$ contains at least one coordinate equal to $1$ and at least one coordinate equal to $0$, such a $j$ always exists for either value of $\trueselect_t$. This ensures $\langle \queryvec_t, \w_t \rangle = v_{t,j} = \trueselect_t$ and $\w_t$ is always a valid database state.

\noindent
\textbf{Step 2: Algorithm's expected loss.~}The learner observes $\queryvec_t$ and then predicts $\predselect_t = \langle \queryvec_t, \predw_t \rangle \in [0, 1]$. Since $\trueselect_t$ is drawn independent of $\predselect_t$, the expected loss per-round is:
\begin{align*}
  \E[\loss_t(\predw_t)] = \E\left[|\trueselect_t - \predselect_t| \right] = \tfrac{1}{2}\predselect_t + \tfrac{1}{2}(1 - \predselect_t) = \tfrac{1}{2}.
\end{align*}
By linearity of expectation, the expected total loss for any algorithm is:
\begin{equation}\label{eq:expected_loss_algorithm}
  \E\big[\sum_{t=1}^T \loss_t(\predw_t)\big] = \frac{T}{2}.
\end{equation}
  
\noindent
\textbf{Step 3: Best fixed predictor's expected loss.~}Consider any fixed weight vector $\mathbf{w}^* \in \Delta_k$. Its absolute loss in round $t$ is $\loss_t(\mathbf{w}^*) = |\trueselect_t - \langle \mathbf{v}_t, \mathbf{w}^* \rangle|$. Since coordinates of $\wstar$ sum to 1, we can rewrite the scalar $\trueselect_t$ as: 
\[
  \trueselect_t = \sum_{i=1}^k w^*_i \trueselect_t.
\]

Substituting this and the inner product $\langle \mathbf{v}_t, \mathbf{w}^* \rangle = \sum_{i=1}^k w^*_i v_{t,i}$ into the absolute loss yields:
\begin{align*}
  \loss_t(\mathbf{w}^*) &= \left| \sum_{i=1}^k w^*_i \trueselect_t - \sum_{i=1}^k w^*_i v_{t,i} \right| = \left| \sum_{i=1}^k w^*_i (\trueselect_t - v_{t,i}) \right|.
\end{align*}
Since $\trueselect_t, v_{t,i} \in \{0,1\}$, the differences $(\trueselect_t - v_{t,i})$ are non-negative for all $i \in [k]$ when $\trueselect_t = 1$, or non-positive when $\trueselect_t = 0$. Because $w^*_i \ge 0$, all terms $w^*_i(\trueselect_t - v_{t,i})$ share the same sign across all $i \in [k]$, so the absolute value of the sum equals the sum of the absolute values. Rewriting the loss at round $t$ as:
\[
  \loss_t(\wstar) = \sum_{i=1}^k w^*_i |\trueselect_t - v_{t,i}|.
% \left| \sum_{i=1}^k w^*_i (\trueselect_t - v_{t,i}) \right| = 
\]
  
Consequently, the total loss of a fixed $\wstar$ can be written as follows:
\[
  \sum_{t=1}^T \loss_t(\mathbf{w}^*) = \sum_{i=1}^k w^*_i \left( \sum_{t=1}^T |\trueselect_t - v_{t,i}| \right).
\]

    Define $S_i := \sum_{t=1}^T |\trueselect_t - v_{t,i}|$ for each $i \in [k]$. 
The total loss is the linear function $\sum_{i=1}^k w^*_i S_i$ over 
$\mathbf{w}^* \in \Delta_k$. A linear function on the simplex attains its 
minimum at a vertex, so:
\begin{equation}\label{eq:sum_loss_best_fixed_predictor_query}
  \min_{\mathbf{w}^* \in \Delta_k} \sum_{t=1}^T \loss_t(\mathbf{w}^*) 
  = \min_{i \in [k]} S_i 
  = \min_{i \in [k]} \sum_{t=1}^T |\trueselect_t - v_{t,i}|.
\end{equation}

Then we rewrite the total loss by defining $z_{t,i} := 1 - 2|\trueselect_t - v_{t,i}|$. Substituting $|\trueselect_t - v_{t,i}| = \frac{1}{2}(1 - z_{t,i})$ into \eqref{eq:sum_loss_best_fixed_predictor_query} and take the expectation. The expected total loss of the best fixed predictor is:

\begin{equation}\label{eq:expected_loss_best_fixed_predictor}
  \begin{aligned}
    &\mathbb{E}\left[\min_{\mathbf{w}^* \in \Delta_k} \sum_{t=1}^T \loss_t(\mathbf{w}^*)\right] 
    = \mathbb{E}\left[\min_{i \in [k]} \sum_{t=1}^T |\trueselect_t - v_{t,i}|\right]\\ 
    =& \mathbb{E}\left[\min_{i \in [k]} \sum_{t=1}^T \frac{1 - z_{t,i}}{2}\right] 
    = \frac{T}{2} - \frac{1}{2}\mathbb{E}\left[\max_{i \in [k]} \sum_{t=1}^T z_{t,i}\right].
  \end{aligned}
  \end{equation}

\noindent
\textbf{Step 4: Reducing the expected regret to independent Rademacher sums.~}
Combining \eqref{eq:expected_loss_best_fixed_predictor} with \eqref{eq:expected_loss_algorithm} and since the maximum over all $k$ coordinates is at least the maximum over the first $k-1$ coordinates, the expected regret for any algorithm is:
\begin{align*}
  \mathbb{E} \left[ \sum_{t=1}^T \loss_t(\predw_t) \right] - \mathbb{E} \left[ \min_{\mathbf{w}^* \in \Delta_k} \sum_{t=1}^T \loss_t(\mathbf{w}^*) \right] = \frac{1}{2}\,\mathbb{E}\left[\max_{i \in [k]} \sum_{t=1}^T z_{t,i}\right]  \ge \frac{1}{2} \mathbb{E} \left[ \max_{i \in [k-1]} \sum_{t=1}^T z_{t,i} \right].
\end{align*}

% \begin{equation}\label{eq:expected_regret_to_independent_rademacher_sums}
%   \mathbb{E} \left[ \regret_T(\mathcal{A}) \right] = \frac{1}{2}\,\mathbb{E}\left[\max_{i \in [k]} \sum_{t=1}^T z_{t,i}\right].
% \end{equation}
Recall that for each coordinate $i\in [k-1]$ and each round $t \in [T]$, both $\trueselect_t$ and $v_{t,i}$ are independent random variables drawn from $\mathrm{Bernoulli}(1/2)$. Therefore, $z_{t,i} = 1 - 2|\trueselect_t - v_{t,i}|$ is independent across $i \in [k-1]$ and $t \in [T]$, and $z_{t,i} \in \{-1, +1\}$ each with probability $1/2$, i.e., a Rademacher random variable. By the standard result as shown in the following \Cref{lem:rademacher_max}~\cite{cesa2006prediction}, the expected maximum of the sum of i.i.d. Rademacher random variables is $\Omega(\sqrt{T \log k})$.
Therefore, for every deterministic algorithm, the expected regret under this randomized adversary is $\Omega(\sqrt{T \log k})$ for sufficiently large $T$ and $k$. By Yao's minimax principle, this implies the same lower bound on the worst-case expected regret of every randomized algorithm.
\end{proof}

The proof relies on two standard results from~\cite{cesa2006prediction}, restated below.

\begin{lemma}[{\cite[Lemma~A.11]{cesa2006prediction}}]
\label{lem:cbl_rademacher_to_gaussian}
Let $\{Z_{i,t}\}_{i \in [N],\, t \ge 1}$ be i.i.d.\ Rademacher
random variables (i.e., $\Prob(Z_{i,t} = -1) = \Prob(Z_{i,t} = 1) = 1/2$),
and let $G_1, \ldots, G_N$ be independent standard normal random
variables. Then
\[
  \lim_{n \to \infty}
  \E\!\left[\max_{i \in [N]} \frac{1}{\sqrt{n}}\sum_{t=1}^n Z_{i,t}\right]
  = \E\!\left[\max_{i \in [N]} G_i\right].
\]
\end{lemma}

\begin{lemma}[{\cite[Lemma~A.12]{cesa2006prediction}}]
\label{lem:gaussian_extremal_asymptotic}
Let $G_1, \ldots, G_N$ be independent standard normal random
variables. Then
\[
  \lim_{N \to \infty}
  \frac{\E[\max_{i \in [N]} G_i]}{\sqrt{2 \ln N}} = 1.
\]
\end{lemma}

\begin{lemma}[Expected maximum of Rademacher sums~\cite{cesa2006prediction}]
\label{lem:rademacher_max}
Let $\{z_{\t,i}\}_{\t \in [T], i \in [k-1]}$ be i.i.d.\ Rademacher random variables
(i.e., $z_{\t,i} \in \{-1, +1\}$ each with probability $1/2$). Then for
sufficiently large $T$ and $k$,
\begin{align*} %\label{eq:rademacher_max}
\E\Bigl[\max_{i \in [k-1]} \sum_{\t=1}^T z_{\t,i}\Bigr] = \Omega(\sqrt{T \log k}).
\end{align*}
\end{lemma}

\begin{proof}
First, by the central limit theorem (\Cref{lem:cbl_rademacher_to_gaussian}), $\sum_{\t=1}^T z_{\t,i} / \sqrt{T}$ converges in distribution to a standard normal $G_i$ as $T \to \infty$:
\begin{align*}
  \lim_{T \to \infty} \E\Bigl[\max_{i \in [k-1]} \frac{1}{\sqrt{T}}\sum_{\t=1}^T z_{\t,i}\Bigr]
  = \E\Bigl[\max_{i \in [k-1]} G_i\Bigr],
\end{align*}
where $G_1,\ldots,G_{k-1}$ are i.i.d.\ standard normal random variables. Second, the expected maximum of $k-1$ i.i.d.\ standard normal random variables satisfies (\Cref{lem:gaussian_extremal_asymptotic}):
\begin{align*}
\lim_{k \to \infty} \frac{\E\bigl[\max_{i \in [k-1]} G_i\bigr]}{\sqrt{2\ln(k-1)}} = 1.
\end{align*}
Combining them and multiplying by $\sqrt{T}$, we have that for sufficiently large $T$ and $k$:
\begin{align*}
\E\bigl[\max_{i \in [k-1]} \sum_{t=1}^T z_{t,i}\bigr] = \Omega(\sqrt{T \log k}).
\end{align*}
\end{proof}

\cut{
By the standard results that lower bounds the expected maximum of the sum of i.i.d. Rademacher random variables~\cite{cesa2006prediction} (See \Cref{lem:cbl_rademacher_to_gaussian} and \Cref{lem:gaussian_extremal_asymptotic} in Appendix~\ref{app:dynamic_absolute_loss_lower}), we have the following convergence results:
\begin{align*}
  \lim_{T \to \infty} \E\bigl[\max_{i \in [k-1]} \frac{\sum_t z_{t,i}}{\sqrt{T}}\bigr] = \E\bigl[\max_{i \in [k-1]} G_i\bigr].
\end{align*}
where $G_1,\ldots,G_{k-1}$ are i.i.d.\ standard normal random variables. Moreover,
\begin{align*}
  \lim_{k \to \infty} \frac{\E\bigl[\max_{i \in [k-1]} G_i\bigr]}{\sqrt{2\ln(k-1)}} = 1.
\end{align*}

Combining them and multiplying by $\sqrt{T}$, we have that for sufficiently large $T$ and $k$:
\begin{align*}
  \E\bigl[\max_{i \in [k-1]} \sum_{t=1}^T z_{t,i}\bigr] = \Omega(\sqrt{T \log k}).
\end{align*}
Substituting back into \Cref{eq:expected_regret_to_independent_rademacher_sums}, we obtain the lower bound on the expected regret:
\begin{align*}
  \mathbb{E} \left[ \regret_T(\mathcal{A}) \right] = \Omega(\sqrt{T \log k}).
\end{align*}
}

\subsection{Upper Bound: Follow-the-Regularized-Leader (FTRL) Algorithm}\label{subsec:dynamic_absolute_loss_upper_bound}
%We obtain the upper bounds for both point and linear queries through a single algorithmic template: Follow-the-Regularized-Leader (FTRL)~\cite{shalev2012online}, with two different regularization functions. 
The choice of regularization function is dictated by the geometry of the query vectors. Point queries are standard basis vectors and satisfy the strong condition $\|\queryvec_t\|_2 = 1$, which lets us use a squared $\ell_2$ regularization function and obtain $O(\sqrt T)$ regret independent of $k$. Linear queries only satisfy the weaker $\|\queryvec_t\|_\infty \le 1$ (with $\|\queryvec_t\|_2$ as large as $\sqrt k$), for which the matching choice is the negative entropy regularization function, yielding $O(\sqrt{T\log k})$ regret.

Both algorithms follow the generic FTRL procedure as detailed in \Cref{alg:ftrl-abs}: they maintain an estimated database vector $\predw_\t$ in each round, by solving an optimization problem over the cumulative loss and a regularization term. The regularization function $R(\w)$ prevents $\predw_\t$ from changing too drastically between rounds and ensures stability. Then the algorithm outputs the estimated selectivity $\predselect_t =\langle\queryvec_t,\predw_t\rangle$ based on the estimated database vector. 
%\red{add the regularization function for each query class here?}

\begin{algorithm}[!htbp]
  \caption{FTRL for Dynamic Databases under Absolute Loss}
  \label{alg:ftrl-abs}
  \begin{algorithmic}[1]
  \Require Learning rate $\eta > 0$, regularization function $R(\w)$: $\frac{1}{\eta} \sum_{i=1}^k w_i \log w_i$ for linear queries or $\frac{1}{\eta} \sum_{i=1}^k w_i^2$ for point queries
  \State Initialize the estimated database vector: $\widehat{\w}_1 \leftarrow (1/k,\ldots,1/k)$ (uniform distribution)
  \For{$t=1,2,\ldots,T$}
      \State \textbf{Predict selectivity:} receive query $\queryvec_t$
             and output: $\predselect_t \leftarrow \langle\queryvec_t,\widehat{\w}_t\rangle$
      \State \textbf{Compute loss:} receive true selectivity $\trueselect_t$ and incur
             $\loss_\t=|\trueselect_t-\predselect_t|$
      \State \textbf{Update the estimated database vector:} $\widehat{\w}_{t+1}\leftarrow
             \arg\min_{\w\in\Delta_k}\Bigl[\sum_{\tau=1}^{t}\loss_\tau(\w)+R(\w)\Bigr]$, where $\loss_\tau(\w) = |\trueselect_\tau - \langle \queryvec_\tau, \w \rangle|$ denotes the absolute loss incurred by $\w$ on round $\tau$.
  \EndFor
  \end{algorithmic}
  \end{algorithm}

% We now present an algorithm that achieves $O(\sqrt{T \log k})$ regret, nearly matching the lower bound up to a logarithmic factor in $k$. We use the standard Follow-the-Regularized-Leader (FTRL) algorithm~\cite{shalev2012online} in online setting with negative entropy regularization. The standard results~\cite{shalev2012online} of FTRL algorithm gives a upper bound of the regret relative to an arbitrary predictor $\mathbf{u}$, presented in the following \Cref{thm:general_ftrl}.

The standard result of FTRL~\cite{shalev2012online} gives an upper bound on the regret relative to an arbitrary comparator $\mathbf{u}$, restated below:

\begin{theorem}[General FTRL Regret Bound \cite{shalev2012online}]
\label{thm:general_ftrl}
Consider a sequence of convex loss functions $\ell_1, \ldots, \ell_T$, such that each $\ell_t$ is $L_t$-Lipschitz with respect to some norm $\|\cdot\|$. Let $L$ satisfy $\tfrac{1}{T}\sum_{t=1}^{T} L_t^{2}\le L^{2}$. Let the regularization function $R$ be a $\sigma$-strongly convex function with respect to the same norm $\|\cdot\|$ over a convex set $\mathcal{K}$. $\regret_T(\mathbf{u}) = \sum_{t=1}^T \loss_t(\predw_t) - \sum_{t=1}^T \loss_t(\mathbf{u})$ is the regret relative to the comparator $\mathbf{u}$. %of the FTRL algorithm 
Then, for all $\mathbf{u} \in \mathcal{K}$,
\begin{align*}
    \regret_T(\mathbf{u}) = \sum_{t=1}^T \loss_t(\predw_t) - \sum_{t=1}^T \loss_t(\mathbf{u}) \le R(\mathbf{u}) - \min_{\mathbf{v} \in \mathcal{K}} R(\mathbf{v}) + \frac{\T L^2}{\sigma},
\end{align*}
\end{theorem}

To derive the specific regret bound for each query class, we establish three key properties separately in \Cref{subsubsec:point_query_upper} for point queries and \Cref{subsubsec:linear_query_upper} for linear queries: (1) the strong convexity of the regularization function, (2) the Lipschitz continuity of the loss function, and (3) the diameter of the regularization function over the simplex.

\subsubsection{Point Queries}\label{subsubsec:point_query_upper}
For point queries we instantiate \Cref{alg:ftrl-abs} with the squared
$\ell_2$ regularizer
\[
   R(\w) \;=\; \tfrac{1}{\eta}\,\|\w\|_2^{2}
        \;=\; \tfrac{1}{\eta}\sum_{i=1}^{k} w_i^{2},
\]
where $\eta > 0$ is the learning rate.
The following three lemmas instantiate \Cref{thm:general_ftrl}, all
measured in the $\ell_2$ norm.

\begin{lemma}\label{lem:l2-strong-convex}
  $R(\w)=\tfrac{1}{\eta}\|\w\|_2^{2}$ is $\tfrac{2}{\eta}$-strongly
  convex with respect to the $\ell_2$ norm on $\Delta_k$.
\end{lemma}

\begin{proof}
  For a differentiable convex $f$, $f$ is $\sigma$-strongly convex w.r.t.\ a norm $\|\cdot\|_2$ iff its Bregman divergence satisfies, for all $\mathbf{p},\mathbf{q}$,
  \[
    D_f(\mathbf{p},\mathbf{q}) := f(\mathbf{p}) - f(\mathbf{q}) - \langle \nabla f(\mathbf{q}), \mathbf{p}-\mathbf{q} \rangle \;\ge\; \frac{\sigma}{2}\,\|\mathbf{p}-\mathbf{q}\|_2^{2}.
  \]
  For $f(\w) = \|\w\|_2^2 = \sum_{i=1}^\k w_i^2$, we have $\nabla f(\w) = 2\w$. Thus for any $\mathbf{p},\mathbf{q}\in\Delta_\k$:
  \begin{align*}
    D_{f}(\mathbf{p},\mathbf{q})
    &= \|\mathbf{p}\|_2^2 - \|\mathbf{q}\|_2^2 - \langle 2\mathbf{q},\, \mathbf{p}-\mathbf{q} \rangle \\
    &= \sum_i p_i^2 + \sum_i q_i^2 - 2\sum_i p_i q_i \\
    &= \|\mathbf{p}-\mathbf{q}\|_2^2 \;\ge\; \frac{2}{2}\,\|\mathbf{p}-\mathbf{q}\|_2^2,
  \end{align*}
  which is exactly the $2$-strong convexity of $f$ w.r.t.\ $\|\cdot\|_2$.

  Finally, strong convexity scales linearly with positive constants. For $g = \alpha \cdot f$ with $\alpha > 0$, if $f$ is $\sigma$-strongly convex, then:
  \begin{align*}
    D_{g}(\mathbf{p},\mathbf{q})
    &= \alpha f(\mathbf{p}) - \alpha f(\mathbf{q}) - \langle \alpha\nabla f(\mathbf{q}), \mathbf{p}-\mathbf{q} \rangle \\
    &\geq \alpha \frac{\sigma}{2}\,\|\mathbf{p}-\mathbf{q}\|^2,
  \end{align*}
  Thus $g$ is $\alpha\sigma$-strongly convex. Applying this to our case, we get that $R(\w) = \tfrac{1}{\eta} \cdot \|\w\|_2^2$ is $\tfrac{2}{\eta}$-strongly convex w.r.t.\ $\|\cdot\|_2$.
\end{proof}

\begin{lemma}\label{lem:l2-lipschitz}
For any $\trueselect_\t \in [0,1]$ and any point query $\queryvec_\t = \evec_{i_\t}$, the loss function $\loss_\t(\w) = |\trueselect_\t - \langle \queryvec_\t, \w \rangle|$ is $1$-Lipschitz with respect to the $\ell_2$ norm on $\Delta_\k$.
\end{lemma}

\begin{proof}
For any $\w, \w' \in \Delta_\k$:
\begin{align*}
|\loss_\t(\w) - \loss_\t(\w')| &= \left||\trueselect_\t - \langle \queryvec_\t, \w \rangle| - |\trueselect_\t - \langle \queryvec_\t, \w' \rangle|\right|\\
&\leq |\langle \queryvec_\t, \w - \w' \rangle| \quad \text{(triangle inequality)}\\
&\leq \|\queryvec_\t\|_2 \cdot \|\w - \w'\|_2 \quad \text{(Cauchy--Schwarz inequality)}.
\end{align*}
For point queries, $\queryvec_\t = \evec_{i_\t}$, so $\|\queryvec_\t\|_2 = 1$. Therefore:
\begin{align*}
|\loss_\t(\w) - \loss_\t(\w')| \leq \|\w - \w'\|_2.
\end{align*}
This establishes that the loss is $1$-Lipschitz with respect to the $\ell_2$ norm.
\end{proof}

\begin{lemma}\label{lem:l2-diameter}
  For $R(\w)=\tfrac{1}{\eta}\|\w\|_2^{2}$ on $\Delta_k$ and any
  $\mathbf{u}\in\Delta_k$,
  $R(\mathbf{u})-\min_{\w\in\Delta_k} R(\w)\le \tfrac{1}{\eta}$.
\end{lemma}

\begin{proof}
We first determine the range of $\|\w\|_2^2$ over $\Delta_\k$.
For the upper bound, since $w_i \in [0,1]$ and $\sum_i w_i = 1$ for any $\w \in \Delta_\k$:
\begin{align*}
\|\w\|_2^2 = \sum_{i=1}^{\k} w_i^2 \leq \max_{i} w_i \cdot \sum_{i=1}^{\k} w_i \leq 1.
\end{align*}
The maximum $\|\w\|_2^2 = 1$ is achieved at the vertices $\evec_i$ of the simplex.

For the lower bound, by the Cauchy--Schwarz inequality, we have
$\langle\mathbf{1}, \mathbf{w}\rangle^2 \leq \|\mathbf{1}\|_2^2 \cdot \|\mathbf{w}\|_2^2$.

The LHS equals to $1$, and the RHS equals to $\k \|\w\|_2^2$, so we have
$1 \leq \k \|\w\|_2^2$. Therefore:
$$\|\w\|_2^2 \geq 1/\k.$$ The minimum $\|\w\|_2^2 = 1/\k$ is achieved at the uniform distribution $\w = (1/\k, \ldots, 1/\k)$.

Therefore:
\begin{align*}
R(\mathbf{u}) - \min_{\w \in \Delta_\k} R(\w) = \frac{1}{\eta}\left(\|\mathbf{u}\|_2^2 - \min_{\w \in \Delta_\k} \|\w\|_2^2\right) =
\frac{1}{\eta}\left(\|\mathbf{u}\|_2^2 - \frac{1}{\k}\right) \leq \frac{1}{\eta}\left(1 - \frac{1}{\k}\right) \leq \frac{1}{\eta}.
\end{align*}
\end{proof}

\begin{corollary}\label{cor:point-query-upper}
For dynamic databases under absolute loss, using FTRL (\Cref{alg:ftrl-abs}) with squared $\ell_2$ regularization function $R(\w)=\tfrac{1}{\eta}\|\w\|_2^{2}$ and learning rate $\eta=\sqrt{2/T}$, the regret for point queries achieves:
\begin{equation}
\regret_{\T} = \sum_{\t=1}^{\T} \loss_\t(\predw_\t) - \min_{\w \in \Delta_\k} \sum_{\t=1}^{\T} \loss_\t(\w) \leq \sqrt{2\T} = O(\sqrt{\T}).
\end{equation}
\end{corollary}

\begin{proof}
  Substituting $\sigma=2/\eta$ (\Cref{lem:l2-strong-convex}), $L=1$
  (\Cref{lem:l2-lipschitz}), and the diameter bound
  (\Cref{lem:l2-diameter}) at most $\tfrac{1}{\eta}$ into \Cref{thm:general_ftrl}: $\regret_T\le \tfrac{1}{\eta}+\tfrac{\eta T}{2}$. Optimizing over
  $\eta$ yields $\eta=\sqrt{2/T}$ and $\regret_T\le\sqrt{2T}$.
\end{proof}

\subsubsection{Linear Queries}
\label{subsubsec:linear_query_upper}
For linear queries we instantiate \Cref{alg:ftrl-abs} with the negative
entropy regularization function: $R(\w) = \tfrac{1}{\eta}\sum_{i=1}^{k} w_i\log w_i$.
The following three lemmas instantiate \Cref{thm:general_ftrl}, all
measured with respect to the $\ell_1$ norm.

\begin{lemma}\label{lem:entropy-strong-convex}
    The negative entropy regularizer $R(\w) = \frac{1}{\eta} \sum_{i=1}^k w_i \log w_i$ is $\frac{1}{\eta}$-strongly convex with respect to the $\ell_1$ norm on $\Delta_k$.
\end{lemma}

\begin{proof}
    For a differentiable convex $f$, $f$ is $\sigma$-strongly convex w.r.t. a norm $\|\cdot\|$ iff its Bregman divergence satisfies, for all $\mathbf{p},\mathbf{q}$,
    \[
      D_f(\mathbf{p},\mathbf{q}) := f(\mathbf{p}) - f(\mathbf{q}) - \langle \nabla f(\mathbf{q}), \mathbf{p}-\mathbf{q} \rangle \;\ge\; \frac{\sigma}{2}\,\|\mathbf{p}-\mathbf{q}\|^{2}.
    \]
    For the negative entropy function $f(\mathbf{w}) = \sum_{i=1}^k w_i \log w_i$, we have $\nabla f(\mathbf{w}) = (1+\log w_i)_{i=1}^{k}$. Thus for any $\mathbf{p},\mathbf{q}\in\Delta_k$, using $\sum_i p_i = \sum_i q_i = 1$:
    \begin{align*}
      D_{f}(\mathbf{p},\mathbf{q})
      &= \sum_i p_i\log p_i - \sum_i q_i\log q_i - \sum_i (1+\log q_i)(p_i-q_i) \\
      &= \sum_i p_i\log p_i - \sum_i p_i\log q_i = \sum_i p_i\log\frac{p_i}{q_i}\\
      &= \mathrm{KL}(\mathbf{p}\,\|\,\mathbf{q}).
    \end{align*}
    By Pinsker's inequality, for all distributions $\mathbf{p},\mathbf{q}\in\Delta_k$,
    \[
      \mathrm{KL}(\mathbf{p}\,\|\,\mathbf{q}) \;\ge\; \tfrac{1}{2}\,\|\mathbf{p}-\mathbf{q}\|_1^{2}.
    \]
    Therefore $D_{f}(\mathbf{p},\mathbf{q}) \ge \tfrac{1}{2}\,\|\mathbf{p}-\mathbf{q}\|_1^{2}$, which is exactly the $1$-strong convexity of $f$ w.r.t. $\|\cdot\|_1$.

    Finally, strong convexity scales linearly with positive constants. For the function $g = \alpha \cdot f$ for any $\alpha > 0$, if $f$ is $\sigma$-strongly convex, then:
    \begin{align*}
      D_{g}(\mathbf{p},\mathbf{q})
      &= \alpha f(\mathbf{p}) - \alpha f(\mathbf{q}) - \langle \alpha\nabla f(\mathbf{q}), \mathbf{p}-\mathbf{q} \rangle \\
      &\geq \alpha \frac{\sigma}{2}\,\|\mathbf{p}-\mathbf{q}\|^2.
    \end{align*}
    Thus $g$ is $\alpha\sigma$-strongly convex. Applying this to our case, we get that $R_\eta = \tfrac{1}{\eta} \cdot \sum_{i=1}^k w_i \log w_i$ is $\tfrac{1}{\eta}$-strongly convex w.r.t. $\|\cdot\|_1$.
\end{proof}

\begin{lemma}\label{lem:abs-lipschitz}
For any $\trueselect_t \in [0,1]$ and $\queryvec_t \in [0,1]^k$, the loss function $\loss_t(\mathbf{w}) = |\trueselect_t - \langle \queryvec_t, \mathbf{w} \rangle|$ is $1$-Lipschitz with respect to the $\ell_1$ norm on $\Delta_k$.
\end{lemma}

\begin{proof}
    For any $\w, \w' \in \Delta_k$:
    \begin{align*}
    |\loss_t(\w) - \loss_t(\w')| &= \left||\trueselect_t - \langle \queryvec_t, \w \rangle| - |\trueselect_t - \langle \queryvec_t, \w' \rangle|\right|\\
    &\leq |\langle \queryvec_t, \w \rangle - \langle \queryvec_t, \w' \rangle| \quad \text{(triangle inequality)}\\
    &= |\langle \queryvec_t, \w - \w' \rangle|\\
    &\leq \|\queryvec_t\|_\infty \cdot \|\w - \w'\|_1 \quad \text{(Hölder's inequality)}\\
    &\leq \|\w - \w'\|_1,
    \end{align*}
    where the last inequality uses $\|\queryvec_t\|_\infty \leq 1$ since $\queryvec_t \in [0,1]^k$.
\end{proof}

\begin{lemma}\label{lem:entropy-diameter}
For the negative entropy regularizer $R(\w) = \frac{1}{\eta} \sum_{i=1}^k w_i \log w_i$ on $\Delta_k$, and any $\mathbf{u} \in \Delta_k$:
$
R(\mathbf{u}) - \min_{\mathbf{w} \in \Delta_k} R(\mathbf{w}) \leq \frac{\log k}{\eta}.
$
\end{lemma}

\begin{proof}
The negative entropy $-\sum_{i=1}^k w_i \log w_i$ is maximized on $\Delta_k$ by
the uniform distribution $\w = (1/k, \ldots, 1/k)$, which gives entropy
$\log k$. The minimum is 0, achieved when the distribution is concentrated on a
single element. Therefore:
\begin{equation}
R(\mathbf{u}) - \min_{\mathbf{w} \in \Delta_k} R(\mathbf{w}) = \frac{1}{\eta}\left[\sum_{i=1}^k u_i \log u_i - \min_{\mathbf{w} \in \Delta_k} \sum_{i=1}^k w_i \log w_i\right] \leq \frac{1}{\eta}(0 - (-\log k)) = \frac{\log k}{\eta}.
\end{equation}
\end{proof}

With these lemmas in place, substituting $\sigma = 1/\eta$, $L = 1$ and applying the diameter bound from Lemma~\ref{lem:entropy-diameter} into \Cref{thm:general_ftrl}, we derive the main upper bound:
\begin{corollary}\label{cor:dynamic-abs-upper}
For dynamic databases under absolute loss, using FTRL (\Cref{alg:ftrl-abs}) with negative entropy regularization function $R(\w) = \tfrac{1}{\eta}\sum_{i=1}^{k} w_i\log w_i$ and learning rate $\eta = \sqrt{\frac{\log k}{T}}$, the regret for linear queries achieves:
\begin{equation}
\regret_{\T} = \sum_{t=1}^T \loss_t(\widehat{\w}_t) - \min_{\w \in \Delta_k} \sum_{t=1}^T \loss_t(\w) \leq 2\sqrt{T \log k} = O(\sqrt{T \log k}) 
\end{equation}
\end{corollary}